\documentclass[11pt,letterpaper]{article}

\usepackage[letterpaper,margin=1in]{geometry}
\usepackage[T1]{fontenc}
\usepackage[utf8]{inputenc}
\usepackage{times}
\usepackage{fullpage}
\usepackage{microtype}
\usepackage{setspace}
\usepackage{amsmath,amssymb,amsthm,mathtools}
\usepackage{xparse}
\usepackage{xcolor}

\usepackage[ruled,vlined,linesnumbered]{algorithm2e}
\SetKw{Break}{break}
\SetKwData{Null}{null}
\SetKwProg{WithProb}{\normalfont with probability}{}{end}

\usepackage{array}
\usepackage{booktabs}
\usepackage{diagbox}
\usepackage{longtable}
\usepackage{multirow}
\usepackage{tabularx}
\usepackage{threeparttable}
\newcolumntype{Y}{>{\raggedright\arraybackslash}X}

\usepackage{enumitem}
\setlist{topsep=0.5em,itemsep=0.2em,parsep=0pt,partopsep=0pt}

\usepackage[
colorlinks=true,
linkcolor=blue!55!black,
citecolor=blue!55!black,
urlcolor=blue!55!black,
pdfauthor={Qin Zhang},
pdftitle={Streaming Algorithms for Gaussian Kernel Density Statistics}
]{hyperref}

\allowdisplaybreaks
\providecommand{\Description}[1]{}

\newtheorem{theorem}{Theorem}[section]
\newtheorem{lemma}[theorem]{Lemma}
\newtheorem{claim}[theorem]{Claim}

\theoremstyle{definition}
\newtheorem{definition}[theorem]{Definition}

\theoremstyle{remark}

\newcommand{\abs}[1]{{\left | #1 \right |}}

\newcommand{\Var}{\mathbf{Var}}
\newcommand{\eps}{\epsilon}

\renewcommand{\Pr}{\mathbf{Pr}}

\newcommand{\R}{\mathbb{R}}
\newcommand{\Z}{\mathbb{Z}}
\newcommand{\E}{\mathbf{E}}

\newcommand{\supp}{\operatorname{supp}}
\newcommand{\poly}{\operatorname{poly}}

\newcommand{\ip}[2]{\left\langle #1,#2\right\rangle}
\newcommand{\norm}[1]{\left\|#1\right\|}
\newcommand{\one}{\mathbf 1}

\newcommand{\calA}{\mathcal{A}}
\newcommand{\packA}{\beta}
\newcommand{\Bgood}{B_{\mathrm{good}}}
\newcommand{\dist}{\operatorname{dist}}

\newcommand{\DI}{{\textsf{DI}}}
\newcommand{\SUM}{{M_1}}

\newcommand{\Index}{\textsc{Index}}

\makeatletter
\let\oldnl\nl
\newcommand{\nonl}{\renewcommand{\nl}{\let\nl\oldnl}}
\makeatother

\title{Streaming Algorithms for Gaussian Kernel Density Statistics}

\author{
	Qin Zhang \\
	Computer Science Department\\
	Indiana University\\
	\texttt{qzhangcs@iu.edu}
}

\date{}

\begin{document}

		\maketitle
		
		\begin{abstract}
			Motivated by data produced by generative systems, \cite{LZ26b} formulates similarity-aware statistics via a weighted similarity graph, replacing equality with similarity in classical frequency-based statistics. Although this framework captures semantic relationships between nonidentical items, under general similarity functions even coarse one-pass approximation can require linear space. We therefore ask whether the geometric structure present in natural vector similarities can overcome this barrier. We answer this question affirmatively for the Gaussian kernel. For fixed-dimensional Euclidean vector streams, we study similarity-aware analogues of classical frequency statistics, including the number of distinct elements and frequency moments, through the diversity index and Gaussian density moments. We give one-pass sublinear-space approximation algorithms that exploit the geometric and analytic properties of the Gaussian kernel, and complement them with lower bounds. Our results show that geometric structure can fundamentally change the streaming complexity of similarity-aware statistical analysis.
		\end{abstract}

\section{Introduction}
\label{sec:intro}

Frequency is meaningful only after specifying when two items should count as the same.  Classical frequency-based statistics make this choice using exact equality: the frequency of an item is the number of dataset items identical to it.  Equivalently, each dataset item contributes either zero or one, according to the equality indicator.  For data produced by modern generative systems, however, exact equality is often too rigid.  Two texts or images may be distinct as data objects while expressing the same concept or content, and their similarity may be real-valued rather than binary.  For
frequency-based statistical analysis of such data, it is therefore natural to replace equality with similarity, allowing each dataset item to contribute according to its similarity to the item of interest.

Formally, let $S=(s_1,\ldots,s_n)\in\mathcal{U}^n$ be a dataset, and let
$sim:\mathcal{U}\times\mathcal{U}\to[0,1]$ be a symmetric similarity function
satisfying $sim(x,x)=1$.  We define the similarity-aware frequency, or weighted degree, of $s_i$ by
$$
D_i = \sum_{j\in[n]} sim(s_i, s_j).
$$
Equivalently, the dataset induces a complete weighted similarity graph with a self-loop at every vertex. The vertices are the input items, and the edge weights are their pairwise similarities. The quantity $D_i$ is the weighted degree of vertex $i$. 

Exact equality is recovered as the special case $sim(x,y)=\mathbf{1}[x=y]$.  Thus, the
degree moments $M_p=\sum_{i\in[n]}D_i^p$ generalize the classical frequency
moments $F_{p+1}=\sum_{a:f_a>0}f_a^{p+1}$, where $f_a$ is the number of
occurrences of item $a$ in the dataset.  The reciprocal-degree sum
$M_{-1}=\sum_{i\in[n]}D_i^{-1}$ similarly provides a similarity-aware analogue
of the number of distinct elements.

We study the computation of these statistics in the standard data stream model~\cite{FM85,AMS99}.  The items arrive sequentially, and the algorithm must
use a small amount of memory to compute a function of the input stream.  The
similarity graph is \emph{not} presented explicitly: only its vertices (the items) arrive in the stream, and an edge weight can be obtained only by evaluating $sim(\cdot, \cdot)$ on two items currently stored in memory.  Consequently, the algorithm can neither store the entire dataset nor explicitly construct the similarity graph.  It must instead estimate global statistics of the implicit graph using a compact streaming sketch.

The study of streaming statistics under general similarity functions was initiated in~\cite{LZ26b}. This generality comes at a significant cost: for general similarities, any one-pass constant-factor approximation to $M_{-1}$
requires $\Omega(n)$ bits of space, and analogous linear-space lower bounds hold for $M_p\ (p \ge 1)$~\cite{LZ26b}. Thus, unlike their classical
frequency-based counterparts, similarity-aware statistics under general similarity functions do {\em not} admit sublinear-space one-pass algorithms. This raises the question of whether structure in a specific, natural
similarity function can overcome the barrier. Our answer is affirmative for the Gaussian kernel: in fixed dimension, we give one-pass sublinear-space algorithms by exploiting analytic and geometric structure specific to the
Gaussian kernel.

Recent work~\cite{Zhang26} gave a complementary affirmative answer for cosine similarity, where low-rank Gram structure leads to one-pass algorithms whose space depends primarily on the representation dimension rather than on the stream length. Our results show that low-rank algebraic structure is {\em not} the only mechanism that can overcome the general-similarity barriers: the spatial decay and Euclidean geometry of the Gaussian kernel provide a different route to sublinear space. 

\vspace{2mm}
\noindent{\bf Gaussian kernel similarity and density statistics.\ }
Let $d$ be a fixed dimension, and let
$P=(p_1,\ldots,p_n)$ be a stream of $n$ points in $\mathbb R^d$.  For a
bandwidth parameter $\sigma>0$, the Gaussian kernel is
$$
K_\sigma(x,y) := \exp\left(-\frac{\norm{x-y}_2^2}{2\sigma^2}\right).
$$
The (unnormalized) Gaussian kernel density of a stream point $p_i$ is defined as
$
D_i^{K_\sigma}:=\sum_{j\in[n]}K_\sigma(p_i,p_j).
$
Because each point has unit similarity to itself, $1\le D_i^{K_\sigma}\le n$.  

We study three fundamental statistics.
\begin{itemize}
	\item {\em Diversity index\ }:
	$
	M_{-1}^{K_\sigma}(P) :=\sum_{i\in[n]}\left(D_i^{K_\sigma}\right)^{-1}.
	$
	To be consistent with prior work~\cite{LZ26b}, we denote $M_{-1}^{K_\sigma}$ by $\DI^{K_\sigma}$. Intuitively, the diversity index discounts multiplicity and measures
	the effective number of distinct regions represented by the stream.
	
	\item {\em First density moment\ }:
	$
	M_1^{K_\sigma}(P) := \sum_{i\in[n]}D_i^{K_\sigma}.
	$
	This quantity can be viewed as a soft self-similarity join size.  It measures the total amount of pairwise similarity in the stream.
	
	\item {\em Higher density moments\ }:
	$
	M_p^{K_\sigma}(P):=\sum_{i\in[n]}\left(D_i^{K_\sigma}\right)^p
	$
	for a fixed integer $p>1$. Higher density moments place greater weight on dense regions.
\end{itemize}
When the kernel is clear from context, we omit the superscript $K_\sigma$ and write the corresponding statistics simply as $D_i$, $\DI$ (or $M_{-1}$), $M_1$, and $M_p$.

\vspace{2mm}
\noindent{\bf Conventions.\ }
We consider insertion-only streams and treat the dimension $d$ as a fixed constant. When considering the higher moments $M_p$, we assume that $p>1$ is a fixed integer.

For our upper bounds, space is measured in words; each
coordinate of an input point and each scalar quantity occupies one word; hence a point in $\mathbb{R}^d$ occupies $O(d)=O(1)$ words. We assume access to independent random reals and truly random
hash functions. These idealizations can be removed using finite-precision discretization and sufficiently independent hash families, with the resulting
errors absorbed into the approximation and failure-probability budgets. Our lower bounds are measured in bits.

For a nonnegative quantity $X$, we call $\widetilde X$ a $(1+\eps,\delta)$-approximation to $X$ if $(1-\eps)X\le \widetilde X\le(1+\eps)X$ with probability at least $1-\delta$, and, for $\alpha\ge1$, an $\alpha$-approximation if $X/\alpha\le \widetilde X\le\alpha X$; we use the former for relative-error guarantees and the latter for constant-factor or coarser ones. We drop $\delta$ when it is clear from context, and call an algorithm
$\delta$-error if it succeeds with probability at least $1-\delta$. We assume $0<\eps<1/2$ throughout, without loss of generality: for $\eps\ge1/2$ one may run the algorithm with parameter $1/2$, at only a constant-factor cost in space.

We omit floor and ceiling operators when they do not affect the asymptotic complexity, and, unless otherwise specified, $\log$ denotes $\log_2$. Constants hidden in $O_d(\cdot)$ may depend on $d$, while those hidden in $O_{p,d}(\cdot)$ may depend on both $p$ and $d$.

\vspace{2mm}
\noindent{\bf Our results.\ }
In this paper, we give the following results.  
\begin{itemize}
	\item \emph{Diversity index.}
	We give a one-pass algorithm that outputs a $(1+\eps,\delta)$-approximation to $\DI(P)$ using
	$
	O_d\left(\eps^{-2d-4}\log^{O(d)}\frac{n}{\eps\delta}\right)
	$
	words of space; see Theorem~\ref{thm:DI}.  For fixed $d$, $\eps$, and
	$\delta$, the dependence on the stream length is only polylogarithmic. 
	We complement this upper bound by establishing an $\Omega_d\left(\log^{d/2} n\right)$-bit space lower bound for one-pass $(1.01, 1/3)$-approximation to $\DI(P)$; see Theorem~\ref{thm:DI-lb}.
	\smallskip
	
	\item \emph{First density moment.}
	We give a one-pass streaming algorithm that outputs a
	$(1+\eps,\delta)$-approximation to $\SUM(P)$ using
	$
	O\left(\eps^{-2}\log\frac{1}{\delta}\right)
	$
	words of space; see Theorem~\ref{thm:SS}.  We complement this upper bound by a lower bound: there is an absolute constant $c_0>0$ such that, for every fixed $\gamma>0$ and every $n^{-1/2+\gamma}\le\eps\le c_0$, any one-pass $1/3$-error algorithm that $(1+\eps)$-approximates $\SUM(P)$ requires $\Omega(\eps^{-2}\log n)$ bits of space; see Theorem~\ref{thm:SS-lb}.
	
	\smallskip
	
	\item \emph{Higher density moments.}
	For every fixed integer $p>1$, we give a one-pass streaming
	algorithm that outputs a $(1+\eps,\delta)$-approximation to $M_p(P)$ using
	$
	O_{p,d}\left(\eps^{-2}n^{1-1/(p+1)}\log n\cdot\log\frac{1}{\delta}\right)
	$
	words of space; see Theorem~\ref{thm:Fp}.  
	We complement this upper bound by a lower bound stating that, for every $C_p n^{-1/(p+1)} \le \eps \le 1/20$ where $C_p$ is a sufficiently large constant, any one-pass $1/3$-error algorithm that $(1+\eps)$-approximates $M_p(P)$ requires $\Omega_p(\eps^{-2} n^{1-2/(p+1)}/\log n)$ bits of space; see Theorem~\ref{thm:Fp-lb}.
\end{itemize}

Table~\ref{tab:general-vs-gaussian} gives a comparison between general
similarity functions and the Gaussian kernel.  The comparison shows that geometric
structure changes the landscape of one-pass space complexity.

\begin{table}[t]
	
	\renewcommand{\arraystretch}{1.18}
	\begin{tabular}{@{}p{0.23\linewidth} p{0.34\linewidth} p{0.37\linewidth}@{}}
		\toprule
		\textbf{Problem}
		& \textbf{General similarity functions}
		& \textbf{Gaussian kernel (this paper)} \\
		\midrule
		Diversity index $\DI$
		& An $O(1)$-approximation requires $\Omega(n)$ bits in one pass~\cite{LZ26b}.
		& One-pass $(1+\eps)$-approximation using
		$O_d\bigl(\eps^{-2d-4}\log^{O(d)}(n/\eps)\bigr)$ words. \\
		
		First density moment $M_1$
		& An $O(1)$-approximation requires $\Omega(n)$ bits in one pass~\cite{LZ26b}.
		& One-pass $(1+\eps)$-approximation using $O(\eps^{-2})$ words. \\
		
		$p$-th density moment, $M_p\ (p>1)$
		& An $O\bigl(n^{(p-1)/2}\bigr)$-approximation requires
		$\Omega(n)$ bits in one pass~\cite{LZ26b}.
		& One-pass $(1+\eps)$-approximation using
		$O_{p,d}\bigl(\eps^{-2}n^{1-1/(p+1)}\log n\bigr)$ words. \\
		\bottomrule
	\end{tabular}
	
	\caption{Streaming complexity for general similarity functions versus
		the Gaussian kernel.  The Gaussian bounds assume a fixed dimension $d$.  The success probability is set to $0.99$ throughout.}
	\label{tab:general-vs-gaussian}
\end{table}

\subsection{Related Work}

\paragraph{Classical streaming statistics.}
Distinct elements ($F_0$) and frequency moments ($F_p$) are central problems in data stream
algorithms (e.g.,~\cite{FM85,AMS99,BJKST02,KNW10}; see~\cite{Muthu05} for a
survey). As noted above, under the equality kernel, the diversity index, first
density moment, and $p$-th density moment reduce to $F_0$, $F_2$, and
$F_{p+1}$, respectively. Under the Gaussian kernel, however, similarity varies
smoothly with distance, so the corresponding density statistics depend on the full geometric
density structure of the point set rather than on discrete frequency classes.

\paragraph{Near-duplicate streams and general similarities.}
Prior work on streaming data with near-duplicates has studied robust distinct
elements and distinct sampling in low-dimensional Euclidean or structured
metric spaces~\cite{CZ16,CZ18,Zhang25,LZ26a}. These works mainly consider
settings in which the data are well-shaped and the induced threshold
similarity graph is close to a disjoint union of cliques. More recent
work~\cite{LZ26b} studied statistics of an implicit similarity graph under general similarity function. Closest to ours is~\cite{Zhang26}, which considered cosine similarity and obtained one-pass algorithms whose space depends primarily on the representation dimension, exploiting the low-rank Gram structure of normalized vectors. Our work specializes this framework to the
Gaussian kernel and investigates how its analytic and geometric structure
changes the one-pass streaming complexity.

\paragraph{Kernel sums and density estimation.}
Kernel density estimation and kernel-sum approximation are fundamental tools
in machine learning, computational geometry, and databases. A major line of
work develops compact coresets that approximate the kernel density uniformly
over all query points. Phillips and Tai established near-optimal coreset
bounds for broad classes of positive-definite kernels, and Tai subsequently
obtained an optimal-size coreset for Gaussian KDE in fixed
dimension~\cite{PT2020,Tai2022}. Charikar et al.\ developed high-dimensional
data structures for relative-error KDE based on density-constrained
near-neighbor search~\cite{CKNS2020}, while Coleman and Shrivastava designed
sketches for pointwise KDE queries~\cite{CS20}. These works approximate
kernel-density queries, either uniformly or at specified query points. In
contrast, we directly estimate global aggregates of the kernel densities
evaluated at the streamed data points.

\paragraph{Similarity-sensitive diversity measures.}
The diversity index is closely related to similarity-sensitive diversity
measures. Let $u$ be the uniform distribution on $P$, and let
$Z_{ij}=K_\sigma(p_i,p_j)$. The order-zero diversity of Leinster and Cobbold
satisfies
$
{}^0D^Z(u) = \sum_{i\in[n]} \frac{u_i}{(Zu)_i} = \sum_{i\in[n]} \frac{1}{D_i} = \DI(P),
$
so our diversity index is precisely their order-zero
similarity-sensitive diversity~\cite{LC12}. Under the equality kernel, this
quantity reduces to the distinct elements $F_0$. For a symmetric
Boolean similarity function, it becomes the Caro--Wei quantity
$\sum_v 1/(\deg(v)+1)$~\cite{Caro1979,Wei1981}. The latter has been studied in edge-arrival and
explicit vertex-arrival graph streams~\cite{CDK18}.\footnote{In the explicit vertex-arrival graph stream model, vertices arrive one by one along with their set of incident edges to previously arrived vertices.} In contrast, our graph is an
implicit complete weighted graph induced by a stream of points through the
Gaussian kernel.

\section{Diversity Index}
\label{sec:DI}

In this section, we give a one-pass upper bound and a dimension-dependent bit-space lower bound for the diversity index.

\begin{theorem}
	\label{thm:DI}
	There is a one-pass streaming algorithm that, given a stream $P=(p_1,\ldots,p_n)$ of $n$ points in $\mathbb R^d$ for a fixed dimension $d$, outputs a $(1+\eps,\delta)$-approximation to $\DI(P)$.
	The algorithm uses $O_d\left(\eps^{-2d-4}\log^{O(d)}\frac{n}{\eps\delta}\right)$ words of space.
\end{theorem}

%\vspace{2mm}
\noindent{\bf Algorithmic overview.\ }
We first discretize $\mathbb R^d$ by a sufficiently fine grid.  For an occupied
cell $g$ with center $z_g$, let
$D_g=\sum_{j\in[n]}K_\sigma(z_g,p_j)$ and let $n_g$ be the number of stream
points in $g$.  The grid resolution ensures that $D_g=(1\pm O(\eps))D_i$ for
every $p_i\in g$, reducing the problem to estimating
$
\DI_\Delta(P):=\sum_{g\in G_{\mathrm{occ}}}{n_g}/{D_g}.
$
To estimate the reciprocal $1/D_g$, we use an {\em exponential lower envelope}.  If
$E_1,\ldots,E_n$ are independent unit-rate exponential random variables, then
$
\min_{j\in[n]}{E_j}/{K_\sigma(z_g,p_j)}
$
is exponential with rate $D_g$ and has expectation $1/D_g$.  Truncating
these envelopes makes every stream update local.  

We call a grid cell {\em active} if its center lies within the truncation radius of some stream point; these are precisely the cells that can receive a nonnegligible local update.
The key to our analysis is a geometric packing lemma showing that the number of active cells is at most $\packA\cdot\DI(P)$, where
$
\packA=O_d\left(\eps^{-d}\log^{2d}\frac{n}{\eps\delta}\right).
$
We can therefore subsample active cells and combine independent estimates using a
median-of-means procedure, leading to the bound in Theorem~\ref{thm:DI}.

\medskip
For ease of presentation, we assume that the stream length $n$ is known in advance. This assumption is not essential: it suffices to know any polynomial upper bound $N\ge n$, which can be used in place of $n$ in all algorithm parameters without changing the stated asymptotic space bound, since $n$ appears only through logarithmic factors.

\subsection{Grid and Cell Densities}

We begin with some geometric properties of the Gaussian kernel.
Let $\Gamma := C_\Gamma \log\left(\frac{n}{\eps\delta}\right)$ for a sufficiently large constant $C_\Gamma=C_\Gamma(d)$. Let $\Delta := \frac{\eps\sigma}{C_\Delta\sqrt\Gamma},$
where $C_\Delta=C_\Delta(d)$ is another sufficiently large constant. 
We impose the fixed axis-aligned grid
$$
G:=\left\{g_k:k\in\mathbb{Z}^d\right\}, \quad \text{where} \quad g_k:=\Delta\bigl(k+[0,1)^d\bigr).
$$
Thus, each cell $g_k$ has side length $\Delta$ and center
$
z_{g_k}:=\Delta\left(k+\frac{1}{2}\mathbf{1}\right).
$
We identify each cell $g_k$ with its index $k\in\mathbb{Z}^d$ whenever
the cell is used as a key for hashing or sampling.

Let $\rho:=\sqrt d\,\Delta$ be the diameter of each cell. For each cell $g$, let $P_g:=\{i:p_i\in g\}$ be the set of indices of points contained in $g$, and let $n_g:=\abs{P_g}$ be the number of points in $g$.  Let $G_{\mathrm{occ}}:=\{g:n_g>0\}$ be the set of occupied cells.  For each cell $g$, define its cell-center Gaussian density by 
$$D_g := \sum_{j \in [n]} K_\sigma(z_g,p_j).$$

The following lemma relates the density at a point to the density at the center of its grid cell. Intuitively, it shows that, when the grid is sufficiently fine, the cell-center density provides a good approximation to the density of every point in the cell.

\begin{lemma}
	\label{lem:center-density}
	For $C_\Delta$ sufficiently large, for every occupied cell $g$ and every $i\in P_g$,
	$
	\left(1-\frac{\eps}{20}\right)D_i \le D_g \le \left(1+\frac{\eps}{20}\right)D_i.
	$
	Consequently, if we define $\DI_\Delta(P) := \sum_{g \in G_{\mathrm{occ}}}\frac{n_g}{D_g}$, then
	$
	\DI_\Delta(P) = \left(1 \pm \frac{\eps}{10}\right)\DI(P).
	$
\end{lemma}

\begin{proof}
	Fix $i\in P_g$. Then $\norm{z_g-p_i}_2\le \rho$. Define $R_0:=\sigma\sqrt{2\ln \left(\frac{40n}{\eps}\right)}$.
	
	For every $p_j$ with $\norm{p_i-p_j}_2\le R_0+\rho$,
	\begin{eqnarray*}
		\abs{\norm{z_g-p_j}_2^2-\norm{p_i-p_j}_2^2}
		&=& \abs{\norm{(z_g-p_i)+(p_i-p_j)}_2^2-\norm{p_i-p_j}_2^2} \\
		&\le& \norm{z_g-p_i}_2^2+2\norm{z_g-p_i}_2\norm{p_i-p_j}_2 \\
		&\le& 2\rho(R_0+\rho)+\rho^2.
	\end{eqnarray*}
	With $C_\Delta$ large enough, $2\rho(R_0+\rho)+\rho^2$ is at most $c_\Delta \eps\sigma^2$ for a sufficiently small absolute constant $c_\Delta$. For such a $j$, compare the two Gaussian terms by taking their ratio:
	$$
	\frac{K_\sigma(z_g,p_j)}{K_\sigma(p_i,p_j)} = \exp \left( -\frac{\norm{z_g-p_j}_2^2-\norm{p_i-p_j}_2^2}{2\sigma^2}\right).
	$$
	The exponent has absolute value at most $c_\Delta \eps/2$, and by taking $c_\Delta$ sufficiently small, equivalently $C_\Delta$ sufficiently large, we have
	$$
	1-\frac{\eps}{40} \le \frac{K_\sigma(z_g,p_j)}{K_\sigma(p_i,p_j)} \le 1+\frac{\eps}{40}.
	$$
	Hence, for all such $j$, $K_\sigma(z_g,p_j)=(1\pm\eps/40)K_\sigma(p_i,p_j)$.
	
	For every $p_j$ with $\norm{p_i-p_j}_2>R_0+\rho$, we have $\norm{p_i-p_j}_2>R_0$ and $\norm{z_g-p_j}_2>R_0$. Therefore, $K_\sigma(p_i,p_j)\le \eps/(40n)$ and $K_\sigma(z_g,p_j)\le \eps/(40n)$.
	
	Since $D_i\ge 1$, the upper bound follows from
	\begin{eqnarray*}
		D_g
		&=& \sum_{j:\norm{p_i-p_j}_2\le R_0+\rho} K_\sigma(z_g,p_j)
		+ \sum_{j:\norm{p_i-p_j}_2>R_0+\rho} K_\sigma(z_g,p_j) \\
		&\le& \left(1+\frac{\eps}{40}\right)\sum_{j:\norm{p_i-p_j}_2\le R_0+\rho} K_\sigma(p_i,p_j)+\frac{\eps}{40} \\
		&\le& \quad  \left(1+\frac{\eps}{20}\right)D_i.
	\end{eqnarray*}
	Similarly,
	\begin{equation*}
		D_g
		\ge \left(1-\frac{\eps}{40}\right)\sum_{j:\norm{p_i-p_j}_2\le R_0+\rho} K_\sigma(p_i,p_j) 
		\ge \left(1-\frac{\eps}{40}\right)\left(D_i-\frac{\eps}{40}\right) 
		\ge \left(1-\frac{\eps}{20}\right)D_i,
	\end{equation*}
	where the last step uses $D_i\ge 1$ and $\eps<1/2$. This proves the first claim.
	
	For the second claim, the first claim implies that for every $i\in P_g$,
	$
	\frac{1}{D_g}=\left(1\pm\frac{\eps}{10}\right)\frac{1}{D_i},
	$
	where we again use $\eps<1/2$. Summing over all occupied cells gives
	$$
	\sum_{g\in G_{\mathrm{occ}}}\frac{n_g}{D_g} = \sum_{g\in G_{\mathrm{occ}}}\sum_{i\in P_g}\frac{1}{D_g} = \left(1\pm\frac{\eps}{10}\right)\sum_{i \in [n]} \frac{1}{D_i}.
	$$
\end{proof}

We shall also use the following immediate consequences.

\begin{lemma}
	\label{lem:cell-density-lower}
	If a cell $g$ is occupied, then $D_g\ge n_g/2$. Moreover, if $i\in P_g$, then $D_i\le 2D_g$.
\end{lemma}

\begin{proof}
	For the first claim, since $g$ is occupied, every point $p_i\in P_g$ satisfies $\norm{z_g-p_i}_2\le \rho$. Hence
	$$
	D_g \ge \sum_{i\in P_g} K_\sigma(z_g,p_i) \ge \exp \left(-\frac{\rho^2}{2\sigma^2}\right)n_g \ge \frac12 n_g,
	$$
	where the last inequality follows from the choice of the grid diameter $\rho$.
	
	For the second claim, Lemma~\ref{lem:center-density} gives $D_g\ge (1-\eps/20)D_i\ge D_i/2$ for every $i\in P_g$. Therefore, $D_i\le 2D_g$.
\end{proof}

\subsection{A Packing Bound for Active Cells}

For the Gaussian kernel, it suffices to consider only cells within a logarithmic radius of a stream point, because contributions from all remaining cells are negligible.

\begin{definition}[Active cell set]
	\label{def:active-cell}
	Let $R^*:=\sigma\sqrt\Gamma$, and define the active cell set as
	$$
	\calA:=\{g\in G:\dist(z_g,P)\le R^*\},
	$$
	where $\dist(z_g,P)=\min_{p\in P}\norm{z_g-p}_2$.
\end{definition}

The following lemma bounds the number of active cells in terms of $\DI(P)$. This bound is the key structural fact underlying both the algorithm and its analysis.

\begin{lemma}[Active-cell packing]
	\label{lem:active-cell}
	Let $\packA=C_\beta\eps^{-d}\Gamma^{2d}$ for a sufficiently large constant $C_\beta=C_\beta(d)$. We have $\abs{\calA}\le \packA\cdot \DI(P)$.
\end{lemma}

\begin{proof}
	Partition $\mathbb{R}^d$ into axis-aligned boxes of side length $R^*$, and let $B$ be the set of boxes containing at least one input point. For a box $b\in B$, let $n_b$ be its number of points. Identifying boxes with their integer grid indices, define the neighboring-box set
	$$
	\mathcal N(b):=\{b'\in B:\norm{b-b'}_\infty\le 1\},
	\quad \text{and} \quad
	n_{\mathcal{N}(b)}:=\sum_{b'\in\mathcal N(b)} n_{b'}.
	$$
	The number of boxes in $\mathcal N(b)$ is at most $3^d$.
	
	If $p_i\in b$, then points outside $\mathcal N(b)$ are at distance at least $R^*$ from $p_i$. Hence, their total contribution to $D_i$ is at most $n \cdot \exp\left(-\frac{(R^*)^2}{2\sigma^2}\right)=ne^{-\Gamma/2}\le 1$
	for a large enough constant $C_\Gamma$. Therefore, $D_i\le n_{\mathcal{N}(b)}+1\le 2n_{\mathcal{N}(b)}$. Thus
	\begin{equation}
		\DI(P) \ge \frac12\sum_{b\in B}\frac{n_b}{n_{\mathcal{N}(b)}}.
		\label{eq:hk-box-lower}
	\end{equation}
	
	Call an occupied box $b$ {\em good} if every occupied neighboring box $b'\in\mathcal N(b)$ satisfies $n_{b'}\le 2n_b$, and {\em bad} otherwise.  Let $\Bgood\subseteq B$ denote the set of good boxes.  If $b$ is good, then $n_{\mathcal{N}(b)}\le 2\cdot 3^d n_b$, so $n_b/n_{\mathcal{N}(b)}\ge 1/(2\cdot 3^d)$. Restricting the sum in \eqref{eq:hk-box-lower} to good boxes gives
	$$
	\DI(P) \ge \frac12\sum_{b\in\Bgood}\frac{n_b}{n_{\mathcal{N}(b)}} \ge \frac12\cdot\abs{\Bgood}\cdot\frac{1}{2\cdot 3^d} = \frac{\abs{\Bgood}}{4\cdot 3^d}.
	$$
	Consequently,
	\begin{equation}
		\abs{\Bgood}\le 4\cdot 3^d\DI(P).
		\label{eq:good-box-count}
	\end{equation}
	
	Every bad box has a neighboring occupied box with more than twice its occupancy. Iterating this step, the occupancies double at each move, so the chain reaches a good box after at most $\log_2 n$ moves. Hence every occupied box lies within $\ell_\infty$ box-distance $O(\log n)$ of a good box. Therefore,
	$
	\abs{B}\le O_d\left(\log^d n \cdot \abs{\Bgood}\right).
	$
	Combining this with \eqref{eq:good-box-count} gives
	\begin{equation}
		\abs{B}\le O_d\left(\log^d n \cdot \DI(P)\right).
		\label{eq:box-count}
	\end{equation}
	
	Now pass from occupied $R^*$-boxes to active fine-grid cells. If a point lies in a box $b$, then every active cell caused by that point has center within distance $R^*$ of the point. Thus each occupied box can create active cells only within an $O_d(R^*)$-neighborhood of the box. The number of such grid cells is at most
	$
	O_d\left(\left(\frac{R^*}{\Delta}\right)^d\right) =	O_d\left(\left(\frac{\Gamma}{\eps}\right)^d\right).
	$
	Using \eqref{eq:box-count} and the fact that $\log n=O(\Gamma)$, we have
	$$
	\abs{\calA} \le O_d\left(\left(\frac{\Gamma}{\eps}\right)^d\Gamma^d\DI(P)\right) = O_d\left(\eps^{-d}\Gamma^{2d}\DI(P)\right).
	$$
	The lemma follows by taking $C_\beta$ sufficiently large.
\end{proof}

\subsection{Exponential Lower Envelopes and Local Updates}

We estimate $\DI_\Delta(P)=\sum_{g\in G_{\mathrm{occ}}}n_g/D_g$, where $D_g=\sum_{j=1}^n K_\sigma(z_g,p_j)$. Recall that by Lemma~\ref{lem:center-density}, $\DI_\Delta(P)$ is a good approximation to $\DI(P)$.

For each occupied cell $g$, we define a random quantity whose expectation approximates $1/D_g$. Multiplying this quantity by the cell count $n_g$ then gives an estimator for $n_g/D_g$.

Sample $E_1,\ldots,E_n\sim\operatorname{Exp}(1)$ independently, and fix a cell center $z_g$.  We call $E_j$ the \emph{exponential rank} of $p_j$.
By the scaling property of exponential random variables,
$E_j/K_\sigma(z_g,p_j)\sim\operatorname{Exp}(K_\sigma(z_g,p_j))$.

Define the lower envelope at cell $g$ by
$$
Y_g^\circ := \min_{j\in[n]}\frac{E_j}{K_\sigma(z_g,p_j)}.
$$
This is the minimum of independent exponential random variables with rates $K_\sigma(z_g,p_j)$. Hence, for every $t\ge 0$,
\begin{equation*}
	\Pr[Y_g^\circ>t] = \prod_{j=1}^n \Pr\left[\frac{E_j}{K_\sigma(z_g,p_j)}>t
	\right]  = \prod_{j=1}^n e^{-K_\sigma(z_g,p_j)t} = \exp\left(-t\sum_{j=1}^n K_\sigma(z_g,p_j)\right) = e^{-D_g t}.
\end{equation*}
Thus $Y_g^\circ\sim\operatorname{Exp}(D_g)$, and $\E[Y_g^\circ]=1/D_g$. Consequently,
$
\E\left[\sum_{g\in G_{\mathrm{occ}}}n_gY_g^\circ\right] = \sum_{g\in G_{\mathrm{occ}}}\frac{n_g}{D_g} = \DI_\Delta(P).
$

The challenge is that the exact quantity $Y_g^\circ$ cannot be maintained using local streaming updates. In principle, every arriving point has a candidate value for every grid cell, so maintaining all exact minima would be infeasible. We therefore truncate the envelope at $T:=C_T\log(1/\eps)$, where $C_T$ is a sufficiently large absolute constant, and define
$$
Y_g:=\min\{T,Y_g^\circ\},
\quad \text{and} \quad
W:=\sum_{g\in G_{\mathrm{occ}}}n_gY_g.
$$
The truncation makes updates local: once the stored value is capped at $T$, a point $p_j$ can matter for cell $g$ only if the quantity $E_j/K_\sigma(z_g,p_j)$ is below $T$.

The following lemma shows that truncation introduces only a small bias.

\begin{lemma}
	\label{lem:truncation-bias}
	$\E[W]=(1\pm\eps/4)\DI(P)$.
\end{lemma}

\begin{proof}
	For $Y\sim\operatorname{Exp}(\lambda)$,
	$$
	\E[\min\{T,Y\}] = \int_0^T \Pr[Y>t]dt = \int_0^T e^{-\lambda t}dt = \frac{1-e^{-\lambda T}}{\lambda}.
	$$
	Thus truncating loses exactly $1/\lambda-\E[\min\{T,Y\}]=e^{-\lambda T}/\lambda$. By Lemma~\ref{lem:cell-density-lower}, every occupied cell has $D_g\ge n_g/2\ge 1/2$. Hence
	$$
	0 \le \frac1{D_g}-\E[Y_g] = \frac{e^{-D_gT}}{D_g} \le \frac{e^{-T/2}}{D_g}.
	$$
	Multiplying by $n_g$ and summing over occupied cells gives
	\begin{equation}
		0 \le \DI_\Delta(P)-\E[W] \le e^{-T/2}\DI_\Delta(P) \le \frac{\eps}{10}\DI_\Delta(P),
		\label{eq:truncation-bias}
	\end{equation}
	where the last inequality holds by choosing $C_T$ large enough so that $e^{-T/2}\le \eps/10$. Therefore, $(1-\eps/10)\DI_\Delta(P)\le \E[W]\le \DI_\Delta(P)$. The lemma follows from Lemma~\ref{lem:center-density}.
\end{proof}

The next lemma bounds the variance of $W$. The main observation is that the square of the cell sum can be controlled by the number of occupied cells. Since every occupied cell is active, this number is at most the number of active cells. Lemma~\ref{lem:active-cell} then converts this geometric bound into a bound in terms of $\DI(P)$, with a loss of the packing factor $\packA$.

\begin{lemma}
	\label{lem:second-moment}
	$\E\left[\sum_{g\in G_{\mathrm{occ}}}(n_gY_g)^2\right]\le 8\DI(P)$, and $\E[W^2]\le 8\packA\cdot\DI(P)^2$.
\end{lemma}

\begin{proof}
	Since $Y_g\le Y_g^\circ$ and $Y_g^\circ\sim\operatorname{Exp}(D_g)$, we have $\E[Y_g^2]\le \E[(Y_g^\circ)^2]=2/D_g^2$. Therefore,
	$$
	\E\left[\sum_{g\in G_{\mathrm{occ}}}(n_gY_g)^2\right] \le 2\sum_{g\in G_{\mathrm{occ}}}\frac{n_g^2}{D_g^2}.
	$$
	For an occupied cell $g$, Lemma~\ref{lem:cell-density-lower} gives $n_g/D_g\le 2$ and $D_i\le 2D_g$ for every $i\in P_g$. Hence
	$$
	\frac{n_g^2}{D_g^2} = \frac{n_g}{D_g}\cdot\frac{n_g}{D_g} \le 2\cdot\frac{n_g}{D_g} = 2\sum_{i\in P_g}\frac1{D_g} \le 4\sum_{i\in P_g}\frac1{D_i}.
	$$
	Summing over all cells gives $\sum_{g\in G_{\mathrm{occ}}}n_g^2/D_g^2\le 4\DI(P)$, and the first claim follows.
	
	For the second claim, by Cauchy--Schwarz,
	$$
	W^2 = \left(\sum_{g\in G_{\mathrm{occ}}}n_gY_g\right)^2 \le \abs{G_{\mathrm{occ}}}\sum_{g\in G_{\mathrm{occ}}}(n_gY_g)^2.
	$$
	Since every occupied cell has its center within distance $\rho\le R^*$ of an input point, $G_{\mathrm{occ}}\subseteq\calA$. Lemma~\ref{lem:active-cell} gives $\abs{G_{\mathrm{occ}}}\le\abs{\calA}\le\packA\cdot\DI(P)$. Taking expectations and using the first claim gives $\E[W^2]\le 8\packA\cdot\DI(P)^2$.
\end{proof}

\vspace{2mm}
\noindent{\bf Local updates suffice.\ }
We now show that the truncated envelope can be maintained by local updates. Suppose a new point $p_j$ with exponential rank $E_j$ arrives. Its candidate value for cell $g$ is $E_j/K_\sigma(z_g,p_j)$. Because the stored value at $g$ is $Y_g=\min\{T,Y_g^\circ\}\le T$, this point can change the stored value only if $E_j/K_\sigma(z_g,p_j)<T$. Using $K_\sigma(z_g,p_j)=\exp\left(-\frac{\norm{z_g-p_j}_2^2}{2\sigma^2}\right)$, this condition is equivalent to
$
E_j\exp\left(\frac{\norm{z_g-p_j}_2^2}{2\sigma^2}\right)<T,
$
and hence to
\begin{equation}
	\norm{z_g-p_j}_2<\sigma\sqrt{2\ln(T/E_j)}.
	\label{eq:local-update-radius}
\end{equation}
Thus, a point with rank $E_j<T$ only needs to be tested against cells in the ball of radius $\sigma\sqrt{2\ln(T/E_j)}$ around it. If $E_j\ge T$, no cell can be updated.

Note that the radius in \eqref{eq:local-update-radius} is random, since it depends on the random exponential rank $E_j$. In the final algorithm (see Section~\ref{sec:algo}), we will use a median-of-means estimator with a total of $n_{\mathrm{rep}}=O(\packA\eps^{-2}\log(1/\delta))$ repetitions. Accordingly, for each point $p_j$, we generate independent exponential ranks $E_j^{(1)},\ldots,E_j^{(n_{\mathrm{rep}})}$ and use one independent rank for each repetition. The following lemma shows that, with high probability, the fixed active radius $R^*$ suffices for all relevant updates, simultaneously over all stream points and all repetitions used by the algorithm.

\begin{lemma}
	\label{lem:local-update}
	With probability at least $1-\delta/20$, for every input point $p_j$, every repetition $r\in[n_{\mathrm{rep}}]$, and every cell $g$, the value $Y_g^{(r)}$ can change because of $p_j$ only if $\norm{z_g-p_j}_2\le R^*$. Moreover, no cell is updated if $E_j^{(r)}\ge T$.
\end{lemma}

\begin{proof}
	Define $\alpha:=\delta/(20n \cdot n_{\mathrm{rep}})$. For each $r\in[n_{\mathrm{rep}}]$, let $E_i^{(r)}$ be the exponential rank assigned to $p_i$ in repetition $r$. Since $\Pr[E_i^{(r)}<\alpha]\le \alpha$ for an $\operatorname{Exp}(1)$ random variable, a union bound gives
	$$
	\Pr\left[\min_{i\in[n],\ r\in[n_{\mathrm{rep}}]}E_i^{(r)}<\alpha\right] \le n\cdot n_{\mathrm{rep}}\cdot\alpha = \frac{\delta}{20}.
	$$
	If $\min_{i\in[n],\ r\in[n_{\mathrm{rep}}]}E_i^{(r)} \ge  \alpha$, then every relevant update radius is at most $\sigma\sqrt{2\ln(T/\alpha)}$. We choose $C_\Gamma$ in the definition of $\Gamma =C_\Gamma\log(n/(\eps\delta))$ large enough so that $2\ln(T/\alpha)\le \Gamma$. This is possible because
	$
	\log(T/\alpha) = O_d\left(\log\frac{n}{\eps\delta}+\log\Gamma\right).
	$
	Therefore, except with probability at most $\delta/20$, every cell that can ever be affected by a truncated-envelope update is contained in $\calA=\{g\in G:\dist(z_g,P)\le R^*\}$, where $R^*=\sigma\sqrt\Gamma$.
\end{proof}

\subsection{Sampling Active Cells}

The active set $\calA$, whose size is bounded by $\packA\DI(P)$, may still be too large to store explicitly. We thus perform subsampling. Let $U:=\abs{\calA}$. We maintain a standard $F_0$ estimator for $U$, such as the algorithm of~\cite{KNW10}, based on updates to active cells and using randomness independent of all cell-sampling hashes and exponential ranks. Active-cell updates are generated by enumerating all grid cells whose centers are within distance $R^*$ of each arriving point. Using $O(\log(n/\delta))$ words, the $F_0$-estimator maintains a value $\widetilde U$ such that, with probability at least $1-\delta/20$,
\begin{equation}
	\frac{1}{2}U \le \widetilde U\le 2U.
	\label{eq:f0-estimate}
\end{equation}

Let $L_{\max}:=\lceil\log_2 n\rceil+1$ and $\mathcal L:=\{0,1,\ldots,L_{\max}-1\}$. The algorithm runs geometric sampling levels $p_\ell=2^{-\ell}$ for $\ell\in\mathcal L$, and $n_{\mathrm{rep}}$ repetitions for each level. Set
$K_0:=C_K\packA\eps^{-2}\log\left(\frac{n_{\mathrm{rep}}L_{\max}}{\delta}\right)$, where $C_K$ is a sufficiently large constant. For each repetition $r\in[n_{\mathrm{rep}}]$ and level $\ell\in\mathcal L$, we use an independent uniform hash function $h_{r,\ell}:G\to[0,1]$ to sample a cell $g$ when $h_{r,\ell}(g)\le p_\ell$. The dictionary $\mathcal{D}_{r, \ell}$ for each pair $(r,\ell)$ stores sampled active cells, capped at $C_{\mathrm{cap}}K_0$ cells. Each stored cell is represented as a tuple $(g,n_g,Y_g^{(r)})$: the cell index $g$, its count $n_g$, and its current truncated envelope value $Y_g^{(r)}$.

At the end of the stream, we choose a level $\ell^*$ using $\widetilde U$. If $\widetilde U\le K_0$, take $\ell^*=0$, so $p_{\ell^*}=1$. Otherwise, choose the smallest level $\ell^*\in\mathcal L$ satisfying $p_{\ell^*}\le K_0/\widetilde U$; if no such level exists, take $\ell^*=L_{\max}-1$. On the event in~\eqref{eq:f0-estimate}, we have
\begin{equation}
	\min\left\{1,\frac{K_0}{4U}\right\} \le p_{\ell^*} \le \min\left\{1,\frac{4K_0}{U}\right\}.
	\label{eq:sampling-level}
\end{equation}

If the selected level satisfies $p_{\ell^*}=1$, then $U\le 2K_0$ on the event \eqref{eq:f0-estimate}. Otherwise, $p_{\ell^*}<1$, and by \eqref{eq:sampling-level} the expected number of sampled active cells at the selected level is $p_{\ell^*}U\le 4K_0$. By a Chernoff bound, for a sufficiently large constant $C_{\mathrm{cap}}$, the number of sampled active cells exceeds $C_{\mathrm{cap}}K_0$ with probability at most $\delta/(20n_{\mathrm{rep}})$ in any repetition. Taking a union bound over all $n_{\mathrm{rep}}$ repetitions, the sampled set at the selected level $\ell^*$ exceeds the storage cap in some repetition with probability at most $\delta/20$.

For each repetition $r\in[n_{\mathrm{rep}}]$, let $S_{r,\ell^*}$ denote the uncapped set of active cells sampled at the selected level, and define the conceptual estimator
$$
\widehat W_r := \frac1{p_{\ell^*}} \sum_{g\in S_{r,\ell^*}}n_gY_g^{(r)}.
$$
On the event that the selected dictionary does not overflow, the implemented estimator equals this conceptual estimator. Conditional on the selected level and the exponential ranks, $\widehat W_r$ is an unbiased estimator of $W_r:=\sum_{g\in G_{\mathrm{occ}}}n_gY_g^{(r)}$.

\begin{lemma}
	\label{lem:W-variance}
	On the event in~\eqref{eq:f0-estimate}, for each selected level $\ell^*$ and each $r\in[n_{\mathrm{rep}}]$, writing $\mu:=\E[W_r]$, $\E[(\widehat W_r-\mu)^2\mid\ell^*] \le 16\packA \cdot \DI(P)^2$.
\end{lemma}

\begin{proof}
	Fix a repetition $r$ and the selected level $\ell^*$, and write $\mathbf E^{(r)}:=(E_j^{(r)})_{j\in[n]}$ for the vector of exponential ranks. Under this conditioning, the values $Y_g^{(r)}$ and $W_r=\sum_{g\in G_{\mathrm{occ}}}n_gY_g^{(r)}$ are fixed. The only remaining randomness in $\widehat W_r$ is the sampling of cells.
	Let $I_g$ be the indicator that cell $g$ is sampled at level $\ell^*$, and write $\theta := p_{\ell^*}$. Then
	$$
	\widehat W_r=\frac1\theta\sum_{g\in G_{\mathrm{occ}}}I_gn_gY_g^{(r)}.
	$$
	The indicators $I_g$ are independent Bernoulli random variables with mean $\theta$. Hence
	$
	\E[\widehat W_r\mid\ell^*,\mathbf E^{(r)}]=W_r,
	$
	and
	\begin{equation*}
		\Var\left(\widehat W_r\mid \ell^*,\mathbf E^{(r)}\right) \le \frac1\theta\sum_{g\in G_{\mathrm{occ}}}\left(n_gY_g^{(r)}\right)^2.
	\end{equation*}
	If $\theta=1$, this conditional variance is zero. If $\theta<1$, then~\eqref{eq:sampling-level} gives $1/\theta\le 4U/K_0$. Since $U=\abs{\calA}\le\packA\cdot\DI(P)$ by Lemma~\ref{lem:active-cell}, we have
	\begin{eqnarray}
		\E\left[\Var\left(\widehat W_r\mid \ell^*,\mathbf E^{(r)}\right)\mid\ell^*\right]
		&\le&
		\frac{4\packA\cdot\DI(P)}{K_0} \E\left[\sum_{g\in G_{\mathrm{occ}}}\left(n_gY_g^{(r)}\right)^2\right] \nonumber\\
		&\le& \frac{4\packA\cdot\DI(P)}{K_0}\cdot 8\DI(P)
		\le 8\packA\cdot\DI(P)^2,
		\label{eq:Var-W-1}
	\end{eqnarray}
	where the second inequality uses Lemma~\ref{lem:second-moment}, and the last uses $K_0\ge4$.
	
	It remains to account for the randomness of the exponential lower envelope. Because the selected level is independent of the exponential ranks, the conditional law of total variance gives
	\begin{equation*}
		\E\left[\left(\widehat W_r-\E[W_r]\right)^2\mid\ell^*\right] =
		\E\left[\Var\left(\widehat W_r\mid \ell^*,\mathbf E^{(r)}\right)\mid\ell^*\right] +
		\E\left[\left(W_r-\E[W_r]\right)^2\right].
	\end{equation*}
	By Lemma~\ref{lem:second-moment},
	\begin{equation}
		\label{eq:Var-W-2}
		\E\left[\left(W_r-\E[W_r]\right)^2\right] \le \E[W_r^2] \le 8\packA\cdot\DI(P)^2.
	\end{equation}
	Combining \eqref{eq:Var-W-1} and \eqref{eq:Var-W-2} gives the lemma.
\end{proof}

\subsection{Pseudocode of the Streaming Algorithm}
\label{sec:algo}

We now give the complete one-pass implementation of the algorithm. Because the description is lengthy, we split the algorithm into three subroutines. The initialization subroutine (Algorithm~\ref{alg:DI-init}) fixes the grid, repetitions, sampling levels, and capped dictionaries. The update subroutine (Algorithm~\ref{alg:DI-update}) is called once for each arriving stream point. The query subroutine (Algorithm~\ref{alg:DI-query}) selects one sampling level and applies median-of-means to the repetition estimates.

\begin{algorithm}[!t]
	\caption{Initialize the grid sketch}
	\label{alg:DI-init}
	\DontPrintSemicolon
	\KwIn{Stream of $n$ points in $\mathbb{R}^d$ for a fixed dimension $d$, bandwidth $\sigma$, accuracy parameter $\eps$, failure probability $\delta$.}
	\KwOut{Initialized sketch state.}
	
	Set $\Gamma\gets C_\Gamma\log(n/(\eps\delta))$, $R^*\gets\sigma\sqrt\Gamma$, and $\Delta\gets\eps\sigma/(C_\Delta\sqrt\Gamma)$\tcp*[r]{active radius and grid side length}
	
	Set $T\gets C_T\log(1/\eps)$ \tcp*[r]{truncation threshold} 
	
	Set $\packA\gets C_\beta\eps^{-d}\Gamma^{2d}$\tcp*[r]{packing factor}
	
	Set $b_m\gets C_b\packA\eps^{-2}$, $\kappa_m \gets C_\kappa\log(1/\delta)$, and $n_{\mathrm{rep}}\gets b_m\kappa_m$\tcp*[r]{median-of-means block size, number of blocks, and total repetitions per sampling level}
	
	Set $L_{\max}\gets\lceil\log_2 n\rceil+1$, $\mathcal L\gets\{0,1,\ldots,L_{\max}-1\}$, and $p_\ell\gets 2^{-\ell}$ for every $\ell\in\mathcal L$\;
	
	Set $K_0\gets C_K\packA\eps^{-2}\log\left(n_{\mathrm{rep}}L_{\max}/\delta\right)$\tcp*[r]{sampled-cell budget}
	
	Initialize an $F_0$ estimator $\mathsf F_0$ for the active-cell set $\calA$\;
	\For{$r=1,\ldots,n_{\mathrm{rep}}$}{
		\For{$\ell\in\mathcal L$}{
			Draw a random hash function $h_{r,\ell}:G\to[0,1]$\tcp*[r]{cell sampling at level $\ell$}
			Initialize an empty dictionary $\mathcal D_{r,\ell}$ with cap $C_{\mathrm{cap}}K_0$ and overflow flag set to false\;
		}
	}
\end{algorithm}

For a point $x$ and radius $R$, write $\mathsf{Near}(x,R):=\{g\in G:\norm{z_g-x}_2\le R\}$. We also write $g(x)$ for the grid cell containing $x$. For each repetition $r$ and level $\ell$, the capped dictionary $\mathcal D_{r,\ell}$ stores records of the form $(g,n_g,Y_g^{(r)})$, where $g$ is the cell index, $n_g$ is the current number of stream points in $g$, and $Y_g^{(r)}$ is the current truncated lower envelope value for repetition $r$. If an insertion would exceed the cap, the dictionary sets an overflow flag; this event is charged to the failure probability in the analysis.

\begin{algorithm}[!t]
	\caption{Process one stream point}
	\label{alg:DI-update}
	\DontPrintSemicolon
	\KwIn{The $j$-th stream point $x$.}
	
	Let $g_x\gets g(x)$\;
	Draw independent ranks $E_j^{(r)}\sim\operatorname{Exp}(1)$ for all $r\in[n_{\mathrm{rep}}]$\tcp*[r]{one rank per repetition}
	\BlankLine
	\tcp{Create sampled records for active cells caused by $x$}
	
	\ForEach{$g\in\mathsf{Near}(x,R^*)$}{
		Feed $g$ to $\mathsf F_0$\tcp*[r]{support update for $\calA$}
		\For{$r=1,\ldots,n_{\mathrm{rep}}$}{
			\For{$\ell\in\mathcal L$}{
				\If{$h_{r,\ell}(g)\le p_\ell$ and $g\notin\mathcal D_{r,\ell}$}{
					\eIf{$\mathcal D_{r,\ell}$ has fewer than $C_{\mathrm{cap}}K_0$ records}{
						Insert $(g,0,T)$ into $\mathcal D_{r,\ell}$\tcp*[r]{count starts at $0$; envelope is truncated at $T$}
					}{
						Set the overflow flag of $\mathcal D_{r,\ell}$ to true\;
					}
				}
			}
		}
	}
	\BlankLine
	\tcp{Update the count of the point's own cell in every sampled dictionary that stores it}
	\For{$r=1,\ldots,n_{\mathrm{rep}}$}{
		\For{$\ell\in\mathcal L$}{
			\If{$g_x$ is stored in $\mathcal D_{r,\ell}$}{
				Increase the stored count $n_{g_x}$ by $1$\;
			}
		}
	}
	\BlankLine
	\tcp{Apply the local lower envelope updates}
	\For{$r=1,\ldots,n_{\mathrm{rep}}$}{
		\If{$E_j^{(r)}<T$}{
			Set $R_j^{(r)}\gets\sigma\sqrt{2\ln \left(T/E_j^{(r)}\right)}$\tcp*[r]{only cells within this radius can improve}
			\ForEach{$g\in\mathsf{Near}(x,R_j^{(r)})$}{
				\ForEach{$\ell\in\mathcal L$ such that $g$ is stored in $\mathcal D_{r,\ell}$}{
					Update $Y_g^{(r)}\gets\min\{Y_g^{(r)},E_j^{(r)}/K_\sigma(z_g,x)\}$ in $\mathcal D_{r,\ell}$\;
				}
			}
		}
	}
\end{algorithm}

\begin{algorithm}[!t]
	\caption{Query the grid sketch}
	\label{alg:DI-query}
	\DontPrintSemicolon
	\KwOut{A $(1+\eps)$-approximation to $\DI(P)$.}
	
	Query $\mathsf F_0$ to obtain $\widetilde U$\;
	\eIf{$\widetilde U\le K_0$}{
		Set $\ell^*\gets 0$\tcp*[r]{store all active cells at level $0$}
	}{
		\eIf{$\{\ell\in\mathcal L:p_\ell\le K_0/\widetilde U\}$ is nonempty}{
			Set $\ell^*\gets\min\{\ell\in\mathcal L:p_\ell\le K_0/\widetilde U\}$\tcp*[r]{$\Theta(K_0)$ expected sampled cells}
		}{
			Set $\ell^*\gets L_{\max}-1$\tcp*[r]{fallback used only on a failure event}
		}
	}
	\If{some selected dictionary $\mathcal D_{r,\ell^*}$ has its overflow flag set}{
		\KwRet{$0$}\tcp*[r]{this is a low-probability failure event}
	}
	\For{$r=1,\ldots,n_{\mathrm{rep}}$}{
		Set $\widehat W_r\gets p_{\ell^*}^{-1}\sum_{g\in\mathcal D_{r,\ell^*}}n_gY_g^{(r)}$\;
	}
	\For{$s=1,\ldots,\kappa_m$}{
		Set $B_s\gets\{(s-1)b_m+1,\ldots,sb_m\}$\;
		Set $\overline W_s\gets b_m^{-1}\sum_{r\in B_s}\widehat W_r$\;
	}
	\KwRet{$\widehat{\DI} \gets \operatorname{median}\{\overline W_1,\ldots,\overline W_{\kappa_m}\}$}.
\end{algorithm}

In Algorithm~\ref{alg:DI-update}, records for cells within distance $R^*$ of the arriving point are created before the count of the point's own cell is incremented; therefore, a sampled active cell does not miss the first point that makes it occupied. The lower envelope update is performed after this insertion step. If $E_j^{(r)}\ge T$, the point cannot improve any truncated envelope value; otherwise, Lemma~\ref{lem:local-update} shows that, with a good probability, the radius $R_j^{(r)}$ includes every cell that the point can impact.

\subsection{Proof of Theorem~\ref{thm:DI}}
\label{subsec:DI-proof}

\begin{proof}
	By Lemma~\ref{lem:truncation-bias}, each repetition satisfies $\mu:=\E[W_r]=(1\pm\eps/4)\DI(P)$. Let $\mathcal E_0$ be the event in~\eqref{eq:f0-estimate}. Conditional on $\mathcal E_0$ and on the selected level $\ell^*$, Lemma~\ref{lem:W-variance} gives
	\begin{equation}
		\label{eq:Var-Wr}
		\E\left[\left(\widehat W_r-\mu\right)^2\mid\ell^*\right] \le 16\packA\cdot\DI(P)^2.
	\end{equation}
	For $s=1,\ldots,\kappa_m$, define
	$
	\overline W_s:=\frac1{b_m}\sum_{r=(s-1)b_m+1}^{s \cdot b_m}\widehat W_r.
	$
	Conditional on $\ell^*$, the block averages are independent and, by~\eqref{eq:Var-Wr}, satisfy
	$$
	\Var(\overline W_s\mid\ell^*) \le \frac{16\packA}{b_m}\DI(P)^2 = \frac{16\eps^2}{C_b}\DI(P)^2.
	$$
	Choose $C_b$ large enough so that this variance is at most $\eps^2\DI(P)^2/160$. Chebyshev's inequality then gives
	$$
	\Pr\left[\left.\abs{\overline W_s-\mu}>\frac{\eps}{4}\DI(P)\,\right|\ell^*\right] \le \frac{1}{10}
	$$
	for every block $s$.
	
	Let $Z_s:=\mathbf 1\left\{\abs{\overline W_s-\mu}>\frac{\eps}{4}\DI(P)\right\}$ for $s=1,\ldots,\kappa_m$.
	Let $Z:=\sum_{s\in[\kappa_m]}Z_s$. Conditional on $\ell^*$, the variables $Z_s$ are independent and satisfy $\Pr[Z_s=1\mid\ell^*]\le 1/10$. A Chernoff bound therefore gives
	$
	\Pr[Z\ge\kappa_m/2\mid\ell^*]\le e^{-c\kappa_m}
	$
	for an absolute constant $c>0$. Choosing $C_\kappa$ sufficiently large ensures that, conditional on $\mathcal E_0$, the median of the conceptual block averages differs from $\mu$ by at most $\frac{\eps}{4}\DI(P)$ with probability at least $1-\delta/2$.
	
	On the event that no selected dictionary overflows, the implemented block averages equal the conceptual ones. Combining the median guarantee with the failure probabilities in Lemma~\ref{lem:local-update}, \eqref{eq:f0-estimate}, and the storage-cap analysis gives an overall failure probability of at most $\delta$.
	
	It remains to analyze space. For each repetition and each sampling level, the capped dictionary stores at most $O(K_0)$ cell records. There are $L_{\max}$ levels and
	$
	n_{\mathrm{rep}} = O\left(\packA\eps^{-2}\log\frac1\delta\right)
	$
	repetitions. Therefore, the number of stored cell records is
	\begin{equation*}
		O(n_{\mathrm{rep}}K_0L_{\max}) = O\left( \packA^2\eps^{-4}L_{\max}\log\frac1\delta \cdot \log\left(\frac{\packA\eps^{-2}\log(1/\delta)L_{\max}}{\delta}\right)\right).
	\end{equation*}
	Substituting $\packA=C_\beta\eps^{-d}\Gamma^{2d}$ and $\Gamma=C_\Gamma\log(n/(\eps\delta))$ gives $\eps^{-2d-4}\log^{O(d)}\left(\frac{n}{\eps\delta}\right)$ words for fixed $d$, plus lower-order space for the $F_0$ estimator. This proves the theorem.
\end{proof}
\subsection{A Dimension-Dependent Space Lower Bound}
\label{subsec:DI-lower-bound}

We next show an $\Omega_d \left(\log^{d/2} n\right)$ lower bound for constant approximation in every fixed dimension $d$.  Throughout this subsection, we fix the Gaussian bandwidth to $\sigma=1$ and write
$
K(x,y):=\exp\left(-\frac{\norm{x-y}_2^2}{2}\right).
$
The result for any other fixed bandwidth follows by scaling all coordinates.

We use the following fixed-weight variant of the standard one-way
\Index\ problem.  The well-known randomized one-way communication
complexity of \Index is $\Omega(m)$~\cite{KN97}.  The same lower bound holds when Alice's input is promised to have Hamming weight exactly $m/2$; we include a proof in Appendix~\ref{sec:proof-lem-index} for completeness.

\begin{lemma}
	\label{lem:fixed-weight-index}
	For an even $m$, Alice receives a vector $x\in\{0,1\}^m$ satisfying $\norm{x}_1=m/2$, and Bob receives an index $i\in[m]$.  Bob must output $x_i$ after receiving one message from Alice.  Every randomized one-way protocol with error probability at most $1/3$ communicates $\Omega(m)$ bits.
\end{lemma}

In this section, we prove the following theorem.

\begin{theorem}
	\label{thm:DI-lb}
	For every fixed dimension $d\ge 1$, any one-pass streaming algorithm that, on every insertion-only stream $P$ of at most $n$ points in $\mathbb R^d$, outputs a $1.01$-approximation to $\DI(P)$ with probability at least $2/3$ must use at least
	$
	\Omega_d \left(\log^{d/2} n\right)
	$
	bits of space. 
\end{theorem}

\vspace{2mm}
\noindent{\bf Proof intuition.\ }
The reduction encodes a fixed-weight \Index\ instance using well-separated clusters of points, one cluster for each $1$-bit. To query a coordinate, Bob appends a sparse $d$-dimensional lattice cloud around the corresponding cluster center. The cloud contains $q=\Theta_d\left(\log^{d/2} n\right)$ points, while its Gaussian self-density remains bounded by a constant. If the queried bit is $0$, each cloud point contributes a constant amount to the diversity index; if it is $1$, the corresponding cluster adds a constant amount to every cloud point's density, decreasing each reciprocal-density contribution by a constant. This creates an $\Omega(q)$ gap in a statistic of size $O(q)$, while the large separation between cluster centers makes all unintended interactions negligible. Hence, a constant-accuracy estimate reveals the queried bit, and the one-way \Index\ lower bound implies an $\Omega_d\left(\log^{d/2} n\right)$-bit space lower bound.

\begin{proof}
	We reduce the fixed-weight \Index{} problem to constant-accuracy estimation of the Gaussian diversity index.  Set
	$
	\lambda:=\sqrt{2\ln(8d)}
	$
	and consider the lattice $\lambda\mathbb Z^d$.  Its total Gaussian similarity sum with respect to the origin is bounded by an absolute constant:
	\begin{eqnarray*}
		B_d
		&:=& \sum_{z \in \lambda\mathbb Z^d} K(0, z) = \sum_{z \in \lambda\mathbb Z^d}
		\exp\left(-\frac{\norm{z}_2^2}{2}\right) = \left(\sum_{k\in\mathbb Z}
		\exp\left(-\frac{\lambda^2k^2}{2}\right) \right)^d\\
		&=& \left(1+2\sum_{k=1}^{\infty}(8d)^{-k^2}\right)^d\\
		&\le& \left(1+\frac{2}{8d-1}\right)^d \le \exp\left(\frac{2d}{8d-1}\right) <2.
	\end{eqnarray*}
	
	Let $N_c:=\lfloor\sqrt n\rfloor$, $r:=\sqrt{2\ln N_c}$, and $t:=\left\lfloor\frac{r}{\lambda\sqrt d}\right\rfloor$, and define the ``probe cloud'' 
	$$
	\mathcal Q:=\left\{\lambda z: z\in\mathbb Z^d, \ \norm{z}_\infty\le t\right\}.
	$$
	Write $q:=\abs{\mathcal Q}=(2t+1)^d$.  For fixed $d$,
	\begin{equation}
		\label{eq:DI-lb-cloud-size}
		q=\Theta_d(r^d)=\Theta_d\left(\log^{d/2} n\right).
	\end{equation}
	Moreover, every $z\in\mathcal Q$ satisfies
	\begin{equation}
		\label{eq:DI-lb-cloud-radius}
		\norm{z}_2\le\lambda\sqrt d t\le r.
	\end{equation}
	
	For $z\in\mathcal Q$, let
	$
	a_z:=\sum_{y\in\mathcal Q}K(z,y)
	$
	be the density contributed by the cloud itself.  Since $\mathcal Q\subseteq\lambda\mathbb Z^d$,
	\begin{equation}
		\label{eq:DI-lb-cloud-density}
		1\le a_z\le B_d<2.
	\end{equation}
	Now suppose that $N_c$ points are placed at the origin.  Their contribution to the density of $z$ is
	$$
	b_z := N_c \cdot K(0,z) = N_c\exp\left(-\frac{\norm{z}_2^2}{2}\right).
	$$
	By~\eqref{eq:DI-lb-cloud-radius},
	\begin{equation}
		\label{eq:DI-lb-center-mass}
		b_z\ge N_c e^{-r^2/2}=1.
	\end{equation}
	It follows from~\eqref{eq:DI-lb-cloud-density} and~\eqref{eq:DI-lb-center-mass} that the $N_c$ points at the origin decrease the reciprocal-density contribution of every cloud point by a constant amount:
	\begin{equation}
		\label{eq:DI-lb-per-point-gap}
		\frac{1}{a_z}-\frac{1}{a_z+b_z} = \frac{b_z}{a_z(a_z+b_z)} \ge \frac{1}{6}.
	\end{equation}
	
	\vspace{2mm}
	\noindent{\bf The reduction.\ }
	Consider fixed-weight \Index\ with $m:=2q$, so Alice's vector has exactly $q$ ones.  Let $e_1$ be the first standard basis vector of $\mathbb R^d$, set $L:=r+10\sqrt{\ln n}$, and define centers $c_j:=jLe_1$ for every $j\in[m]$.
	
	Given $x\in\{0,1\}^m$ with $\norm{x}_1=q$, Alice inserts $N_c$ copies of $c_j$ for every $j$ satisfying $x_j=1$.  She then sends the streaming algorithm's memory state to Bob.  Given $i\in[m]$, Bob appends the translated cloud
	$$
	c_i+\mathcal Q = \{c_i+z:z\in\mathcal Q\}.
	$$
	The resulting stream contains $q N_c+q=q(N_c+1)$ points.  By~\eqref{eq:DI-lb-cloud-size} and the choice $N_c=\lfloor\sqrt n\rfloor$, this is at most $n$ for a fixed $d$ and sufficiently large $n$.
	
	\vspace{2mm}
	\noindent{\bf Analysis with the idealized instances.\ }
	We first ignore all interactions between distinct Alice clusters, and all interactions between the probe cloud and Alice clusters other than the queried cluster at $c_i$.  Let $\overline V_b$ denote the resulting idealized diversity index when $x_i=b$.
	
	If $x_i=0$, all $q$ Alice clusters are isolated.  Each cluster consists of $N_c$ identical points and therefore contributes exactly one to the diversity index.  Hence
	\begin{equation}
		\label{eq:DI-lb-V0}
		\overline V_0 = q + \sum_{z\in\mathcal Q}\frac{1}{a_z}.
	\end{equation}
	
	If $x_i=1$, there are $q-1$ unaffected Alice clusters.  Let
	$
	S_{\mathcal Q}:=\sum_{z\in\mathcal Q}K(0,z).
	$
	The $N_c$ points at the queried center have density $N_c+S_{\mathcal Q}$ and together contribute $N_c/(N_c+S_{\mathcal Q})$.  The cloud point indexed by $z$ has density $a_z+b_z$.  Therefore,
	\begin{equation}
		\label{eq:DI-lb-V1}
		\overline V_1 = (q-1) + \frac{N_c}{N_c+S_{\mathcal Q}} + \sum_{z\in\mathcal Q}\frac{1}{a_z+b_z}.
	\end{equation}
	Subtracting~\eqref{eq:DI-lb-V1} from~\eqref{eq:DI-lb-V0} and using~\eqref{eq:DI-lb-per-point-gap}, we obtain
	\begin{equation}
		\label{eq:DI-lb-additive-gap}
		\overline V_0-\overline V_1 = 1-\frac{N_c}{N_c+S_{\mathcal Q}} + \sum_{z\in\mathcal Q}
		\left(\frac{1}{a_z}-\frac{1}{a_z+b_z}\right) \ge \frac{q}{6}.
	\end{equation}
	Also, since $a_z\ge1$, we have
	\begin{equation}
		\label{eq:DI-lb-size-bound}
		\overline V_1\le \overline V_0\le 2q.
	\end{equation}
	
	\vspace{2mm}
	\noindent{\bf  The actual instance.\ }
	Two distinct Alice centers are separated by at least $L$.  Every cloud point lies within distance $r$ of $c_i$, so its distance from any Alice center $c_j$ with $j\ne i$ is at least
	$
	L-r=10\sqrt{\ln n}.
	$
	Thus every interaction omitted from the idealized instance has Gaussian weight at most
	$
	\exp\left(-\frac{(L-r)^2}{2}\right)=n^{-50}.
	$
	There are at most $n$ points in the stream, so the total directed Gaussian mass of all omitted interactions is at most $n^{-50} \cdot n^2 = n^{-48}$.
	
	For a point $v$, let $\overline D_v$ be its density in the idealized instance and write its actual density as $D_v=\overline D_v+E_v$, where $E_v\ge0$ is the contribution of the omitted interactions.  Since $\overline D_v,D_v\ge1$,
	$$
	0 \le \frac{1}{\overline D_v}-\frac{1}{D_v} = \frac{E_v}{\overline D_vD_v} \le E_v.
	$$
	Summing over all points shows that, uniformly over all valid \Index{} inputs with $x_i=b$,
	\begin{equation}
		\label{eq:DI-lb-ideal-error}
		\overline V_b-n^{-48} \le \DI(P) \le \overline V_b.
	\end{equation}
	
	Choose $\eps_0:=0.01$.  By~\eqref{eq:DI-lb-additive-gap}, \eqref{eq:DI-lb-size-bound}, and~\eqref{eq:DI-lb-ideal-error},
	\begin{equation*}
		(1-\eps_0)(\overline V_0-n^{-48})-(1+\eps_0)\overline V_1 \ge \frac{q}{6} -\eps_0(\overline V_0+\overline V_1) -n^{-48} \ge \left(\frac{1}{6}-4\eps_0\right)q-n^{-48} >0
	\end{equation*}
	for all sufficiently large $n$.  Consequently, the possible $(1+\eps_0)$-approximation ranges in the cases $x_i=0$ and $x_i=1$ are disjoint.  Bob can recover $x_i$ by comparing the output with any fixed threshold between these two ranges.
	Thus, Lemma~\ref{lem:fixed-weight-index} implies that Alice's memory state must have size at least
	$
	\Omega(m)=\Omega(q)=\Omega_d\left(\log^{d/2} n\right)
	$
	bits, where the last equality follows from~\eqref{eq:DI-lb-cloud-size}.
\end{proof}

\section{First Density Moment}
\label{sec:SS}

In this section, we study the first density moment and establish the following result.

\begin{theorem}
	\label{thm:SS}
	There is a one-pass streaming algorithm that, given a stream $P=(p_1,\ldots,p_n)$ of $n$ points in $\mathbb R^d$ for a fixed dimension $d$, outputs a $(1+\eps,\delta)$-approximation to $\SUM(P)$.
	The algorithm uses
	$
	O\left(\eps^{-2}\log\frac1\delta\right)
	$
	words of space. 
\end{theorem}

\noindent{\bf Algorithmic overview.\ }  The Gaussian kernel has an exact continuous square-root feature identity: there
are functions $q_x\in L_2(\mathbb R^d)$ such that
$\langle q_x,q_y\rangle=K_\sigma(x,y)$.  Consequently, for
$Q=\sum_{i\in[n]}q_{p_i}$,
$
\SUM(P)=\norm{Q}_{L_2(\mathbb R^d)}^2.
$
We discretize and truncate the continuous features to obtain a deterministic
sparse local feature map $\phi$ satisfying a uniform additive kernel error of
$\eta/n$.  If $V=\sum_{i\in[n]}\phi(p_i)$, then $\norm{V}_2^2=(1\pm\eta)\SUM(P)$.  The additive feature error becomes relative
because $\SUM(P)\ge n$, and an AMS second-moment sketch estimates
$\norm{V}_2^2$ in one pass using $O(\eps^{-2}\log\frac{1}{\delta})$ words.

\medskip

Again, for ease of presentation, we assume that the stream length $n$ is known in advance. This assumption is not essential: it suffices to know any polynomial upper bound $N\ge n$, which can be used in place of $n$ in the feature-map parameters without changing the stated asymptotic space bound, since $n$ appears only through logarithmic factors in these parameters.

\vspace{2mm}
\noindent{\bf Continuous feature identity.\ }
The starting point of our algorithm is an exact continuous feature representation of the Gaussian kernel.  For every $x\in\R^d$, define
$$
q_x(z)=\left(\frac{2}{\pi\sigma^2}\right)^{d/4} \exp\left(-\frac{\norm{z-x}_2^2}{\sigma^2}\right),
\qquad \forall z\in\R^d.
$$
This is a square-root Gaussian feature: the inner product of two such functions equals the Gaussian kernel.  For completeness, we prove the identity in Appendix~\ref{sec:Gaussian-identity}.

\begin{lemma}
	\label{lem:SS-continuous-feature}
	For all $x,y\in\R^d$,
	$
	\ip{q_x}{q_y}_{L_2(\R^d)}=K_\sigma(x,y).
	$
	Consequently, if
	$
	Q(z)=\sum_{i \in [n]} q_{p_i}(z),
	$
	then
	$
	\SUM(P)=\norm{Q}_{L_2(\R^d)}^2.
	$
\end{lemma}

\vspace{2mm}
\noindent{\bf A shared local feature grid.\ }
Our algorithm uses a deterministic local feature map for the Gaussian kernel.  This is in the same spirit as explicit finite-dimensional feature approximations for shift-invariant kernels, such as random Fourier features~\cite{RR07}.  The key difference is that our construction is grid-based and local: each point updates only a bounded neighborhood of grid cells, which is what makes the map suitable for streaming updates.

We now define a discretized and truncated version of $q_x$.  Let $\eta=c_\eta\eps$ be a feature-approximation parameter for a sufficiently small constant $c_\eta$.  Set the feature-grid scale
$
\Lambda_\eta=C_\Lambda\log\frac{dn}{\eta},
$
the feature-grid spacing
$
h=\frac{\sigma}{C_h\sqrt{\Lambda_\eta}},
$
and the feature truncation radius
$
L=C_L\sigma\sqrt{\Lambda_\eta}.
$

For $x\in\R^d$, define the local support
$
\mathcal S(x)=\{t\in\Z^d:\norm{ht-x}_\infty\le L\}
$
and the feature vector $\phi(x)\in\ell_2(\Z^d)$ by
$$
\phi(x)_t = h^{d/2} \left(\frac{2}{\pi\sigma^2}\right)^{d/4} \exp\left(-\frac{\norm{ht-x}_2^2}{\sigma^2}\right) \one[t\in\mathcal S(x)], \qquad \forall t\in\Z^d.
$$
The map is nonnegative and local.  Its support size satisfies
$$
\abs{\supp(\phi(x))}=\abs{\mathcal S(x)}\le\left(\frac{2L}{h}+1\right)^d= O_d(\Lambda_\eta^d).
$$

The following lemma shows that the inner product of two feature vectors provides a good approximation to the Gaussian kernel.
\begin{lemma}
	\label{lem:SS-kernel-approx}
	For sufficiently large choices of $C_\Lambda,C_h,C_L$, the local feature map satisfies, for all $x,y\in\R^d$,
	$
	\left|\ip{\phi(x)}{\phi(y)}_{\ell_2(\Z^d)}-K_\sigma(x,y)\right| \le \frac{\eta}{n}.
	$
\end{lemma}

\begin{proof}
	First ignore truncation and define
	$
	\psi(x)_t=h^{d/2}q_x(ht)\ (t\in\Z^d).
	$
	Then
	$$
	\ip{\psi(x)}{\psi(y)} = h^d\sum_{t\in\Z^d}q_x(ht)q_y(ht).
	$$
	The summand is a Gaussian centered at $(x+y)/2$ with integral $K_\sigma(x,y)$.  By Poisson summation,
	\begin{equation*}
		h^d\sum_{t\in\Z^d}q_x(ht)q_y(ht) = K_\sigma(x,y) \sum_{k\in\Z^d} \exp\left(-\frac{\pi^2\sigma^2\norm{k}_2^2}{2h^2}\right) \exp\left(-\frac{2\pi i}{h}k\cdot\frac{x+y}{2}\right).
	\end{equation*}
	The $k=0$ term is exactly $K_\sigma(x,y)$.  The remaining terms have total absolute value at most
	$$
	K_\sigma(x,y) \sum_{k\ne 0} \exp\left(-\frac{\pi^2\sigma^2\norm{k}_2^2}{2h^2}\right) \le \sum_{k\ne 0} \exp\left(-\frac{\pi^2\sigma^2\norm{k}_2^2}{2h^2}\right).
	$$
	Since $h=\sigma/(C_h\sqrt{\Lambda_\eta}) = \sigma/(C_h\sqrt{C_\Lambda \log(dn/\eta)})$, choosing $C_h$ and $C_\Lambda$ large enough (as functions of the fixed dimension $d$) makes this last sum at most $\eta/(3n)$, uniformly in $x$ and $y$.
	
	It remains to account for truncation. We use the following property of Gaussian density; the proof can be found in Appendix~\ref{app:discrete-gaussian-tail}.
	\begin{claim}
		\label{lem:discrete-gaussian-tail}
		For the feature maps $\psi$ and $\phi$ defined above, if
		$L=C_L\sigma\sqrt{\Lambda_\eta}$ and $C_L,C_\Lambda$ are sufficiently large
		as functions of the fixed dimension $d$, then
		$
		\|\psi(u)-\phi(u)\|_2 \leq \frac{\eta}{12n}
		$
		uniformly for all $u\in\mathbb{R}^d$.
	\end{claim}
	
	The same Poisson-summation estimate with $x=y=u$ gives $\norm{\psi(u)}_2\le 2$, and $\phi(u)$ is obtained from $\psi(u)$ by deleting coordinates, so $\norm{\phi(u)}_2\le 2$.  Hence, for arbitrary $x,y$,
	\begin{eqnarray*}
		\left|\ip{\psi(x)}{\psi(y)} - \ip{\phi(x)}{\phi(y)} \right|
		&\le&
		\norm{\psi(x)-\phi(x)}_2\norm{\psi(y)}_2 + \norm{\phi(x)}_2\norm{\psi(y)-\phi(y)}_2 \\
		&\le&
		2\cdot\frac{\eta}{12n}+2\cdot\frac{\eta}{12n} = \frac{\eta}{3n}.
	\end{eqnarray*}
	Combining the discretization and truncation errors gives the claimed bound $\eta/n$.
\end{proof}

For the stream $P$, define the implicit aggregate feature vector
$$
V=\sum_{i \in [n]}\phi(p_i)\in\ell_2(\Z^d).
$$
The vector $V$ is never stored explicitly; the algorithm stores only linear sketches of it.

\begin{lemma}
	\label{lem:SS-feature-consequences}
	Let
	$
	\widetilde D_i=\ip{\phi(p_i)}{V}_{\ell_2(\Z^d)}
	$
	and
	$
	\widetilde{\SUM}=\norm{V}_2^2=\sum_{i \in [n]}\widetilde D_i.
	$
	Then
	$
	\widetilde D_i=(1\pm\eta)D_i
	$
	and
	$
	\widetilde{\SUM}=(1\pm\eta)\SUM(P).
	$
\end{lemma}

\begin{proof}
	For every $i$,
	\begin{equation*}
		\abs{\widetilde D_i-D_i}
		= \biggl|\sum_{j \in [n]} \left(\ip{\phi(p_i)}{\phi(p_j)}-K_\sigma(p_i,p_j)\right) \biggr| 
		\le \sum_{j \in [n]} \left|\ip{\phi(p_i)}{\phi(p_j)}-K_\sigma(p_i,p_j)\right| \le n\cdot\frac{\eta}{n} = \eta.
	\end{equation*}
	Since $D_i\ge 1$, this is a relative $\eta$ error for each density.  Summing over $i$ gives
	$
	\abs{\widetilde{\SUM}-\SUM(P)}\le \eta n.
	$
	Because all diagonal kernel terms (i.e., $K_\sigma(p_i, p_i)$ for $i \in [n]$) equal $1$, $\SUM(P)\ge n$, so this is also a relative $\eta$ error for the first Gaussian density moment.
\end{proof}

We estimate $\SUM(P)$ by applying an AMS second-moment sketch~\cite{AMS99} to the implicit vector $V$.

\begin{lemma}[AMS sketch~\cite{AMS99}]
	\label{lem:SS-ams}
	Let $V\in\ell_2(\Z^d)$ be fixed with finite support, and let $\chi:\Z^d\to\{-1,+1\}$ be a 4-wise independent sign function.  Define an AMS counter
	$
	A:=\sum_{t\in\Z^d}\chi(t)V_t.
	$
	For any $0<\alpha<1$, suppose $b_m=C_b\alpha^{-2}$ and $\kappa_m=C_\kappa\log(1/\delta)$ for sufficiently large constants $C_b$ and $C_\kappa$.  Group $b_m\kappa_m$ independent AMS counters into $\kappa_m$ blocks of size $b_m$, average their squared values within each block, and return the median of the block averages.  This median-of-means estimator is a $(1+\alpha)$-approximation to $\norm{V}_2^2$ with probability at least $1-\delta$.
\end{lemma}

\begin{algorithm}[!t]
	\caption{One-pass AMS sketch for the first Gaussian density moment}
	\label{alg:SS-ams}
	\DontPrintSemicolon
	\KwIn{Stream of $n$ points in $\mathbb{R}^d$ for a fixed dimension $d$, bandwidth $\sigma$, accuracy parameter $\eps$, failure probability $\delta$.}
	\KwOut{A $(1+\eps)$-approximation to $\SUM(P)$.}
	
	Set $\eta\gets c_\eta\eps$ for a sufficiently small absolute constant $c_\eta$\tcp*[r]{feature error}
	
	Set $\Lambda_\eta\gets C_\Lambda\log(dn/\eta)$, $h\gets\sigma/(C_h\sqrt{\Lambda_\eta})$, and $L\gets C_L\sigma\sqrt{\Lambda_\eta}$\tcp*[r]{feature grid scale, spacing, and truncation radius}
	
	Set $a_\sigma\gets (2/(\pi\sigma^2))^{d/4}$\tcp*[r]{normalization}
	
	Set $b_m\gets C_b\eps^{-2}$ and $\kappa_m\gets C_\kappa\log(1/\delta)$\tcp*[r]{median-of-means block size and number of blocks}
	
	Choose mutually independent sign functions $\chi_{s,r}:\Z^d\to\{-1,+1\}$, each drawn from a 4-wise independent family, for $s=1,\ldots,\kappa_m$ and $r=1,\ldots,b_m$\tcp*[r]{AMS signs}
	
	Initialize $A_{s,r}\gets 0$ for all $s,r$\tcp*[r]{AMS counters}
	\For{each stream point $x$}{
		Enumerate $\mathcal S(x)=\{t\in\Z^d:\norm{ht-x}_\infty\le L\}$\tcp*[r]{local support}
		\ForEach{$t\in\mathcal S(x)$}{
			Set $w\gets h^{d/2}a_\sigma\exp(-\norm{ht-x}_2^2/\sigma^2)$\tcp*[r]{$w=\phi(x)_t$}
			\For{$s=1,\ldots,\kappa_m$}{
				\For{$r=1,\ldots,b_m$}{
					$A_{s,r}\gets A_{s,r}+\chi_{s,r}(t)w$
				}
			}
		}
	}
	\For{$s=1,\ldots,\kappa_m$}{
		Set $Z_s\gets b_m^{-1}\sum_{r \in [b_m]}A_{s,r}^2$
	}
	\KwRet{$\widehat{\SUM}\gets\operatorname{median}\{Z_1,\ldots,Z_{\kappa_m}\}$}.
\end{algorithm}

Our algorithm is presented in Algorithm~\ref{alg:SS-ams}.  The algorithm uses a deterministic local feature map to reduce the first density moment to an $\ell_2$-norm estimation problem.  Each point $x$ is mapped to a sparse vector $\phi(x)$ supported on nearby grid coordinates $\mathcal S(x)$, and for $V=\sum_{i\in[n]}\phi(p_i)$ we have $\norm{V}_2^2=(1\pm O(\eps))\SUM(P)$.  The stream is processed by maintaining an AMS sketch of this implicit vector $V$: when $x$ arrives, the algorithm enumerates $\mathcal S(x)$, computes the weights $\phi(x)_t$, and updates the sketch.  At the end, the AMS estimator returns an approximation to $\norm{V}_2^2$, which gives the desired estimate of $\SUM(P)$.

\begin{proof}[Proof of Theorem~\ref{thm:SS}]
	By Lemma~\ref{lem:SS-feature-consequences},
	$
	\norm{V}_2^2=\widetilde{\SUM}=(1\pm\eta)\SUM(P).
	$
	By Lemma~\ref{lem:SS-ams} (setting $\alpha = \eps/4$), the AMS median-of-means estimator in Algorithm~\ref{alg:SS-ams} returns a $(1 + \eps/4)$-approximation to $\norm{V}_2^2$ with probability at least $1-\delta$.  Taking $c_\eta=1/4$, the product of the two approximation factors, one from the feature approximation and one from the AMS sketch, is bounded by $(1\pm \eps)$. Therefore,
	$
	\Pr\left[(1-\eps)\SUM(P) \le \widehat{\SUM} \le (1+\eps)\SUM(P) \right] \ge 1-\delta.
	$
	
	Algorithm~\ref{alg:SS-ams} maintains $b_m\kappa_m=O(\eps^{-2}\log(1/\delta))$ AMS counters, which dominate the space usage.
\end{proof}

\subsection{Lower Bound}
\label{sec:SS-lb}

We next show that estimating the first Gaussian density moment is at least as hard as estimating the classical second frequency moment.  For a stream over a universe $U$ with $\abs{U}=\poly(n)$, let $f_a$ denote the frequency of item $a\in U$, and recall that
$
F_2(f):=\sum_{a\in U}f_a^2.
$
For $\eps \ge n^{-1/2+\gamma}$, the one-pass space complexity of obtaining a $(1+\eps)$-approximation to $F_2$ on a stream of length $n$ is
$
\Omega\left(\eps^{-2}\log n\right)
$
bits~\cite{BZ25}.

\begin{theorem}
	\label{thm:SS-lb}
	There is an absolute constant $c_0>0$ such that, for every fixed constant $\gamma>0$ and all
	$
	n^{-1/2+\gamma}\le\eps\le c_0,
	$
	any one-pass streaming algorithm that, on every stream of at most $n$ points on the real line, returns a $(1+\eps)$-approximation to $\SUM(P)$ with probability at least $2/3$ requires
	$
	\Omega\left(\eps^{-2}\log n\right)
	$
	bits of space.
\end{theorem}

\begin{proof}
	We reduce the insertion-only $F_2$ problem to estimating the first Gaussian density moment.  Let $f=(f_a)_{a\in U}$ be the final frequency vector of the stream, and let $\ell:=\sum_{a\in U}f_a\le n$ be its length.  Fix the Gaussian bandwidth to $\sigma=1$ and write $K:=K_1$.  Fix an arbitrary injective indexing $\pi:U\to\mathbb N$.
	Set $\eta:=\frac{\eps}{100}$, and $L:=\left\lceil\sqrt{2\ln\frac n\eta}\right\rceil$. For every occurrence of item $a\in U$, insert the point
	$$
	z_a:=\pi(a)L\in\mathbb R.
	$$
	Thus, the $f_a$ occurrences of item $a$ become $f_a$ coincident Gaussian points at location $z_a$.  Distinct item locations are separated by at least $L$, and hence, for all $a\ne b$,
	$$
	K(z_a,z_b) = \exp\left(-\frac{|z_a-z_b|^2}{2}\right)\le \exp\left(-\frac{L^2}{2}\right) \le \frac{\eta}{n}.
	$$
	
	The contribution of pairs at the same location is exactly $F_2(f)$.  Therefore,
	$$
	\SUM(P) = \sum_{a,b\in U}f_af_bK(z_a,z_b) = \sum_{a\in U}f_a^2 + \sum_{\substack{a,b\in U\\a\ne b}}f_af_bK(z_a,z_b).
	$$
	The cross-location contribution is nonnegative and is at most
	$$
	\frac{\eta}{n} \sum_{\substack{a,b\in U\\a\ne b}}f_af_b \le \frac{\eta}{n} \left(\sum_{a\in U}f_a\right)^2 = \frac{\eta\ell^2}{n} \le \eta\ell.
	$$
	Since every $f_a$ is a nonnegative integer,
	$
	F_2(f) = \sum_{a\in U}f_a^2 \ge \sum_{a\in U}f_a = \ell.
	$
	Consequently,
	$$
	F_2(f) \le \SUM(P) \le (1+\eta)F_2(f).
	$$
	
	Suppose there is a one-pass algorithm using $s$ bits that returns a $(1+\eps/20)$-approximation $\widehat S$ to $\SUM(P)$.  Then
	$
	(1-\eps/20)F_2(f) \le \widehat S \le (1+\eps/20)(1+\eta)F_2(f).
	$
	For $\eta=\eps/100$, this is a $(1+\eps/10)$-approximation to $F_2(f)$.  The Gaussian algorithm would therefore give a one-pass $F_2$ estimator using the same $s$ bits of space.  The frequency-moment lower bound implies
	$
	s = \Omega\left(\eps^{-2}\log n \right).
	$
	Replacing $\eps/20$ by $\eps$ changes only the constants.
\end{proof}

\section{Higher Density Moments}
\label{sec:density-moments}

The algorithm for the first density moment relies on a linear sketch of a feature-space
aggregate and does not extend directly to $p>1$.  In this section, we give a
different one-pass estimator for higher density moments and establish the following result.

\begin{theorem}
	\label{thm:Fp}
	Let $p>1$ be a fixed integer and $d$ a fixed dimension.  There is a one-pass streaming algorithm that, given a stream $P=(p_1,\ldots,p_n)$ of $n$ points in $\mathbb R^d$, produces a $(1+\eps,\delta)$-approximation to $M_p(P)$.
	The algorithm uses
	$
	O_{p,d}\left(\eps^{-2} n^{1-1/(p+1)}\log n\cdot \log\frac1\delta\right)
	$
	words of space.
\end{theorem}

\noindent{\bf Algorithmic overview.\ }  We use a sample-and-probe approach.  The algorithm samples candidate points uniformly from the stream and maintains $p$ independent sets of probe samples.  For a candidate point $p_i$, each probe set provides an unbiased estimate of $D_i$; multiplying the $p$ independent estimates gives an unbiased estimate of $D_i^p$.  Averaging over candidates and rescaling then gives an unbiased estimator of $M_p(P)$.

The main challenge is variance.  Reusing the same probe samples across many
candidates introduces significant correlation, and a black-box moment argument
is insufficient.  We establish geometric regularity properties of Gaussian
densities: a low-density set cannot contribute too much Gaussian mass near a
fixed query point, a point of high density forces many points to have comparable
density, and Gaussian overlap sums control correlations between different
candidates.  These bounds give a sample size of
$O_{p,d}\left(\eps^{-2}n^{1-1/(p+1)}\log n\right)$ per independent copy and, after median amplification, imply Theorem~\ref{thm:Fp}.

\medskip

We will first assume that the stream length is known in advance.  We then explain in Section~\ref{sec:unknown-length} how to remove this assumption without changing the estimator or its asymptotic space complexity.

\subsection{Geometric Regularity of Gaussian Densities}
\label{sec:geo-Fp}

We begin by establishing several properties of density moments in fixed-dimensional Gaussian geometry.  For $x\in\mathbb R^d$, write
$
D(x):=\sum_{j\in[n]}K_\sigma(x,p_j),
$
so that $D(p_i)=D_i$.  The following lemma states that a set of low-density input points cannot place too much Gaussian mass near any fixed query point.

\begin{lemma}
	\label{lem:low-density-packing}
	Let $Q$ be a subset of stream points satisfying $D(q)\le\theta$ for every
	$q\in Q$.  Then for every $x\in\mathbb R^d$,
	$
	\sum_{q\in Q}K_\sigma(x,q)\le C_Q \theta
	$
	for a sufficiently large constant $C_Q=C_Q(d)$.
\end{lemma}

\begin{proof}
	For $t=0,1,2,\ldots$, define $A_t:=\{q\in Q:t\sigma\le \|q-x\|_2<(t+1)\sigma\}$ to be the $t$-th annulus of points around $x$.
	
	We first bound $|A_t|$.  The ball $B(x,(t+1)\sigma)$ can be covered by
	$C'_Q(t+1)^d$ Euclidean balls of radius $c\sigma$, for a sufficiently small absolute
	constant $c>0$ and a sufficiently large constant $C'_Q = C'_Q(d)$.  Consider one covering ball $B$ with $B\cap A_t\neq\emptyset$, and fix
	$q_0\in B\cap A_t$.  For every $q\in B\cap A_t$, we have $\|q-q_0\|_2\le 2c\sigma$
	and $K_\sigma(q_0,q)\ge e^{-2c^2}$. Since $D(q_0)\le \theta$, we have
	$$
	|B\cap A_t| \cdot e^{-2c^2} \le \sum_{q\in B\cap A_t} K_\sigma(q_0,q) \le D(q_0) \le \theta.
	$$
	Thus $|B\cap A_t|\le e^{2c^2}\theta$, and summing over the covering balls gives
	$$
	|A_t|\le C'_Q(t+1)^d e^{2c^2}\theta.
	$$
	For $q\in A_t$, we have $K_\sigma(x,q)\le e^{-t^2/2}$.  Therefore
	$$
	\sum_{q\in Q}K_\sigma(x,q) \le \sum_{t\ge0}|A_t|e^{-t^2/2} \le C'_Q e^{2c^2}\theta \sum_{t\ge0}(t+1)^d e^{-t^2/2} \le C_Q\theta,
	$$
	where $\sum_{t\ge0}(t+1)^d e^{-t^2/2}$ converges for fixed $d$.
\end{proof}

The next lemma says that a point of large Gaussian density cannot be isolated: a large
value of $D_i$ forces many input points to have comparable density.

\begin{lemma}
	\label{lem:density-self-regularity}
	For every stream point $p_i$,
	$
	M_p(P) \ge c_K D_i^{p+1}
	$
	for a sufficiently small constant $c_K=c_K(p,d)$.
\end{lemma}

\begin{proof}
	Let $\theta=D_i$.  Apply Lemma~\ref{lem:low-density-packing} to
	$
	Q:=\{p_j:D_j\le \theta/(2C_Q)\}
	$
	with query point $p_i$.  The total contribution of $Q$ to $D_i$ is at most
	$\theta/2$.  Hence the complementary subset
	$
	H:=\{p_j:D_j>\theta/(2C_Q)\}
	$
	contributes at least $\theta/2$ to $D_i$.  Since each Gaussian kernel value is at most
	$1$, this implies $|H|\ge \theta/2$.  Therefore
	$$
	M_p(P) \ge \sum_{p_j\in H}D_j^p \ge \frac{\theta}{2}\left(\frac{\theta}{2C_Q}\right)^p = c_K \theta^{p+1}
	$$
	for a small enough constant $c_K$.
\end{proof}

Set $\tau:=n^{1/(p+1)}$.
Since $D_i\ge 1$ for every $i$, we have
\begin{equation}
	\label{eq:density-moment-lower-trivial}
	M_p(P)\ge n=\tau^{p+1}.
\end{equation}

The following lemma bounds the power sums of the densities, which will be used in the subsequent variance analysis of our estimator.

\begin{lemma}
	\label{lem:density-power-sums}
	For every integer $q\in\{0,1,\ldots,p\}$,
	$
	\sum_{i\in[n]} D_i^{2p-q} = O_{p,d} \left(\frac{\left(M_p(P)\right)^2 \log n}{\tau^{q+1}} \right).
	$
\end{lemma}

\begin{proof}
	Define the dyadic density levels
	$
	L_a:=\{i:2^a\le D_i<2^{a+1}\}\ \bigl(a=0,1,\ldots,\log n\bigr).
	$
	Since $1\le D_i\le n$, these levels cover all stream points.  
	
	Fix one level $L_a$.  Since $\sum_{i\in L_a}D_i^p\le M_p(P)$, we have
	$|L_a|\le {M_p(P)}/{2^{ap}}$.
	Hence
	\begin{equation}
		\label{eq:density-level-power}
		\sum_{i\in L_a}D_i^{2p-q} = \sum_{i\in L_a}D_i^pD_i^{p-q} \le 2^{(a+1)(p-q)}\sum_{i\in L_a}D_i^p \le C_p M_p(P)\,2^{a(p-q)},
	\end{equation}
	where $C_p=2^p$ is a constant.
	
	We analyze two cases:
	\begin{itemize}
		\item If $2^a\le \tau$, then
		\begin{equation}
			\label{eq:tau-1}
			M_p(P)2^{a(p-q)} \le M_p(P)\tau^{p-q} \le \frac{\left(M_p(P)\right)^2}{\tau^{q+1}},
		\end{equation}
		where the last inequality uses~\eqref{eq:density-moment-lower-trivial}. 
		
		\item If $2^a>\tau$ and
		$L_a$ is nonempty, then Lemma~\ref{lem:density-self-regularity} gives
		$M_p(P)\ge c_K 2^{a(p+1)}$.  Therefore
		\begin{equation}
			\label{eq:tau-2}
			M_p(P)2^{a(p-q)} \le c_K^{-1}\frac{\left(M_p(P)\right)^2}{2^{a(q+1)}} \le
			c_K^{-1}\frac{\left(M_p(P)\right)^2}{\tau^{q+1}}.
		\end{equation}
	\end{itemize}	 
	Combining~\eqref{eq:tau-1} and \eqref{eq:tau-2} with~\eqref{eq:density-level-power} and summing over $O(\log n)$ levels
	proves the lemma.
\end{proof}

For two input points $p_i,p_j$, define their Gaussian overlap
$$
S_{ij}:=\sum_{u\in[n]} K_\sigma(p_i,p_u)K_\sigma(p_j,p_u).
$$
The next lemma is the main geometric correlation bound used in the variance analysis of our algorithm.

\begin{lemma}[Gaussian correlation bound]
	\label{lem:gaussian-correlation-bound}
	For every integer $q\in\{1,2,\ldots,p\}$,
	$$
	\sum_{i,j\in[n]} S_{ij}^q(D_iD_j)^{p-q} \le O_{p,d}\left( \frac{\log n}{\tau^q}{\left(M_p(P)\right)^2}\right).
	$$
\end{lemma}

\begin{proof}
	Expand $S_{ij}^q$ over $q$ probe indices.  For a tuple
	$U=(u_1,\ldots,u_q)\in[n]^q$, define
	$
	K_U(i):=\prod_{a \in [q]} K_\sigma(p_i,p_{u_a}).
	$
	Then
	$$
	S_{ij}^q=\sum_{U\in[n]^q}K_U(i)K_U(j),
	$$
	and hence
	$$
	\sum_{i,j\in[n]} S_{ij}^q(D_iD_j)^{p-q} = \sum_{U\in[n]^q} \left(\sum_{i\in[n]}D_i^{p-q}K_U(i)\right)^2.
	$$
	
	Write
	$
	T_U:=\sum_{i\in[n]}D_i^{p-q}K_U(i).
	$
	First,
	\begin{equation}
		\label{eq:sum-TU}
		\sum_{U\in[n]^q}T_U = \sum_{i\in[n]}D_i^{p-q}\left(\sum_{u\in[n]}K_\sigma(p_i,p_u)\right)^q =
		\sum_{i\in[n]}D_i^{p-q} D_i^q =
		\sum_{i\in[n]}D_i^p = M_p(P).
	\end{equation}
	We next prove an upper bound on $T_U$.  Since each kernel value is at most $1$,
	$
	K_U(i)\le K_\sigma(p_i,p_{u_1}).
	$
	Consider the dyadic density level $L_a=\{i:2^a\le D_i<2^{a+1}\}$.  Applying
	Lemma~\ref{lem:low-density-packing} to the set $Q=\{p_i:i\in L_a\}$ with
	$\theta=2^{a+1}$ and query point $p_{u_1}$ gives
	$$
	\sum_{i\in L_a}K_\sigma(p_i,p_{u_1})\le C_Q 2^{a+1}.
	$$
	Therefore, the contribution of level $L_a$ to $T_U$ is at most (recall that $C_p=2^p$)
	\begin{equation}
		\label{eq:TU-level}
		\left(2^{a+1}\right)^{p-q}\cdot C_Q 2^{a+1} \le C_p C_Q 2^{a(p-q+1)}. 
	\end{equation}
	We again analyze two cases:
	\begin{enumerate}
		\item If $2^a\le \tau$, then using~\eqref{eq:density-moment-lower-trivial} gives
		\begin{equation}
			\label{eq:tau-3}
			2^{a(p-q+1)} \le \tau^{p-q+1} \le \frac{M_p(P)}{\tau^q}.
		\end{equation}
		
		\item  If $2^a>\tau$ and $L_a$ is nonempty,
		then Lemma~\ref{lem:density-self-regularity} gives
		$M_p(P)\ge c_K 2^{a(p+1)}$, and therefore
		\begin{equation}
			\label{eq:tau-4}
			2^{a(p-q+1)} \le c_K^{-1}\frac{M_p(P)}{2^{aq}} \le c_K^{-1}\frac{M_p(P)}{\tau^q}.
		\end{equation}
	\end{enumerate}
	Combining~\eqref{eq:tau-3} and \eqref{eq:tau-4} with~\eqref{eq:TU-level} and summing over $O(\log n)$ levels gives
	\begin{equation}
		\label{eq:TU-uniform}
		T_U\le C_p C_Q c_K^{-1}\log n\cdot \frac{M_p(P)}{\tau^q},
	\end{equation}
	for every $U\in[n]^q$.
	Finally, by~\eqref{eq:sum-TU} and~\eqref{eq:TU-uniform},
	$$
	\sum_{U\in[n]^q}T_U^2 \le \left(\max_U T_U \right) \sum_{U\in[n]^q}T_U \le C_p C_Q c_K^{-1}\log n\cdot \frac{\left(M_p(P)\right)^2}{\tau^q},
	$$
	which proves the lemma.
\end{proof}

\subsection{Algorithm and Analysis}
\label{sec:algo-Fp}

Our algorithm is presented in Algorithm~\ref{alg:density-moments}.  At a high level, it samples candidate points uniformly from the stream, uses independent probe samples to estimate their densities, and averages the resulting unbiased estimates of $D_i^p$.

\begin{algorithm}[!t]
	\caption{One-pass density-moment estimator}
	\label{alg:density-moments}
	\DontPrintSemicolon
	\KwIn{Stream of $n$ points in $\mathbb R^d$ for a fixed dimension $d$, bandwidth $\sigma$, integer moment parameter $p>1$, accuracy parameter $\eps$, and failure probability $\delta$.}
	\KwOut{A $(1+\eps)$-approximation to $M_p(P)$.}
	
	Set $\tau\gets n^{1/(p+1)}$, $b_m\gets C_b\eps^{-2}\tau^p\log n$, and $\kappa_m\gets C_\kappa\log(1/\delta)$ \tcp*[r]{density scale, sample size per repetition, and number of repetitions}
	
	\For{$r=1,\ldots,\kappa_m$}{
		Initialize candidate samplers $S^{(r)}_1,\ldots,S^{(r)}_{b_m}\gets\bot$\;
		Initialize probe samplers $R^{(r)}_{\ell,k}\gets\bot$ for all $\ell\in[p]$ and $k\in[b_m]$\;
	}
	\BlankLine
	\tcp{Streaming: maintain independent uniform samples with replacement}
	\ForEach{$j$-th stream point $p_j$}{
		\For{$r=1,\ldots,\kappa_m$}{
			\For{$a=1,\ldots,b_m$}{
				With probability $1/j$, set $S^{(r)}_a\gets p_j$\;
			}
			\For{$\ell=1,\ldots,p$}{
				\For{$k=1,\ldots,b_m$}{
					With probability $1/j$, set $R^{(r)}_{\ell,k}\gets p_j$\;
				}
			}
		}
	}
	\BlankLine
	\tcp{Post-processing on the stored candidate samples and probe samples}
	\For{$r=1,\ldots,\kappa_m$}{
		\For{$a=1,\ldots,b_m$}{
			\For{$\ell=1,\ldots,p$}{
				Set $\widehat D^{(r)}_{\ell,a}\gets \frac{n}{b_m}\sum_{k \in [b_m]} K_\sigma(S^{(r)}_a,R^{(r)}_{\ell,k})$\;
			}
			Set $Z^{(r)}_a\gets\prod_{\ell \in [p]}\widehat D^{(r)}_{\ell,a}$\;
		}
		Set $\widehat M^{(r)}\gets \frac{n}{b_m}\sum_{a \in [b_m]} Z^{(r)}_a$\;
	}
	\KwRet{$\widehat M_p \gets\operatorname{median}\{\widehat M^{(1)},\ldots,\widehat M^{(\kappa_m)}\}$}.
\end{algorithm}

Recall that $M_p(P)=\sum_{i\in[n]} D_i^p$ and $D_i=\sum_{j\in[n]} K_\sigma(p_i,p_j)$.
If we could sample a point $x$ uniformly from the stream and compute
$D(x)^p = \left(\sum_{j\in[n]}K_\sigma(x,p_j)\right)^p$ exactly, then $n\cdot D(x)^p$
would be an unbiased estimator for $M_p(P)$.  

During the stream, the algorithm maintains $b_m$ independent candidate samplers and, for each $\ell\in[p]$, $b_m$ independent probe samplers. At the end of the stream, the stored points in all samplers are independent uniform samples from the stream, with replacement. Let $S$ denote the multiset of candidate samples, and let $R_\ell$ denote the multiset of samples in the $\ell$-th probe set. For each sampled candidate $x$ and each $\ell\in[p]$, we form the density estimate
$$
\widehat D_\ell(x) = \frac{n}{b_m} \sum_{y\in R_\ell} K_\sigma(x,y).
$$
Each $\widehat D_\ell(x)$ is an unbiased estimator for $D(x)$.  Moreover,
the samples $R_1,\ldots,R_p$ are independent, and therefore
$
\prod_{\ell \in [p]} \widehat D_\ell(x)
$
is an unbiased estimator for $D(x)^p$.  The final estimator is
$$
\widehat M_p = \frac{n}{b_m} \sum_{x\in S} \prod_{\ell \in [p]} \widehat D_\ell(x).
$$

Note that during the stream, we only maintain reservoir samples.  All kernel evaluations are performed after the stream and only between stored candidate and probe points.  

The main technical issue is that the same probe samples are reused for all candidate points, which creates correlations among the summands.  The fixed-dimensional Gaussian geometry is used to control these correlations: roughly, a point of high Gaussian density cannot be isolated, and the low-density regions cannot contribute too much Gaussian mass to any query point.  This gives a density-level correlation bound and a variance bound for the
shared-probe estimator.

For the analysis, it suffices to analyze one repetition.  Dropping the repetition superscript $r$, let
$I_1,\ldots,I_{b_m}$ be the candidate indices sampled by $S_1,\ldots,S_{b_m}$, and let
$R_{\ell,1},\ldots,R_{\ell,b_m}$ denote the probe indices in the $\ell$-th probe set.  These indices are independent and uniform in $[n]$.  For every stored candidate $p_i$ and every $\ell\in[p]$, define
$$
\widehat D_\ell(p_i) := \frac{n}{b_m}\sum_{a \in [b_m]} K_\sigma(p_i,p_{R_{\ell,a}}).
$$
The corresponding base estimator is
\begin{equation}
	\label{eq:density-moment-estimator}
	\widehat M_p:= \frac{n}{b_m}\sum_{a \in [b_m]} \prod_{\ell \in [p]} \widehat D_\ell(p_{I_a}).
\end{equation}

\begin{lemma}
	\label{lem:density-moment-unbiased}
	The estimator in~\eqref{eq:density-moment-estimator} satisfies
	$
	\E[\widehat M_p]=M_p(P).
	$
\end{lemma}

\begin{proof}
	For every fixed $i$ and every $\ell$,
	$
	\E[\widehat D_\ell(p_i)] = \frac{n}{b_m}\cdot b_m\cdot \frac{D_i}{n} = D_i.
	$
	The $p$ probe samples are independent, so
	$
	\E\left[\prod_{\ell \in [p]}\widehat D_\ell(p_i)\right]=D_i^p.
	$
	Since each candidate sample is uniform in $[n]$,
	$
	\E[\widehat M_p] = n\cdot \frac1n\sum_{i\in[n]}D_i^p = M_p(P).
	$
\end{proof}

We analyze the variance of $\widehat M_p$.
\begin{lemma}
	\label{lem:density-moment-variance}
	$
	\Var[\widehat M_p] \le \eps^2\left(M_p(P)\right)^2.
	$
\end{lemma}

\begin{proof}
	Let $Z_i:=\prod_{\ell \in [p]}\widehat D_\ell(p_i)$ and $Z:=\sum_{i\in[n]} Z_i$.
	We decompose the variance according to the probe samples $R$ and the candidate samples $I$:
	\begin{equation}
		\label{eq:law-total-variance}
		\Var[\widehat M_p] = \Var_{R}\left[\E_I[\widehat M_p\mid R]\right] + \E_R\left[\Var_I(\widehat M_p\mid R)\right].
	\end{equation}
	Conditioned on $R$,
	$$
	\E_I[\widehat M_p\mid R] = \frac{n}{b_m}\cdot b_m\cdot \frac1n\sum_{i\in[n]}Z_i = Z.
	$$
	Thus the first term in~\eqref{eq:law-total-variance} is $\Var_R[Z]$.
	
	We first bound $\Var_R[Z]$.  For one probe sample, a direct calculation gives
	\begin{equation}
		\label{eq:one-probe-second-moment}
		\E[\widehat D(p_i)\widehat D(p_j)] = D_iD_j+\frac{1}{b_m}\left(nS_{ij}-D_iD_j\right) \le D_iD_j+\frac{n}{b_m}S_{ij}.
	\end{equation}
	Since the $p$ probe samples are independent, we have
	\begin{equation}
		\label{eq:Z}
		\E_R[Z_iZ_j] \le \left(D_iD_j+\frac{n}{b_m}S_{ij}\right)^p.
	\end{equation}
	Since $\E_R[Z]=M_p(P)$ and $\E_R[Z^2] = \sum_{i,j\in[n]}\E_R[Z_iZ_j]$, expanding the right-hand side of \eqref{eq:Z} gives
	\begin{eqnarray*}
		\Var_R[Z]
		&=& \E_R[Z^2]-(\E_R[Z])^2 \\
		&\le& \sum_{q=0}^p \binom{p}{q}
		\left(\frac{n}{b_m}\right)^q
		\sum_{i,j\in[n]}S_{ij}^q(D_iD_j)^{p-q}-\left(M_p(P)\right)^2 \\
		&\le& \sum_{i,j\in[n]}(D_iD_j)^p+C'_Z\sum_{q=1}^p
		\left(\frac{n}{b_m}\right)^q
		\sum_{i,j\in[n]}S_{ij}^q(D_iD_j)^{p-q}-\left(M_p(P)\right)^2 \\
		&=& C'_Z\sum_{q=1}^p
		\left(\frac{n}{b_m}\right)^q
		\sum_{i,j\in[n]}S_{ij}^q(D_iD_j)^{p-q}.
	\end{eqnarray*}
	Here, $C'_Z=C'_Z(p)$ is a sufficiently large constant. The last equality holds because
	$
	\sum_{i,j\in[n]}(D_iD_j)^p = \left(\sum_{i\in[n]}D_i^p\right)^2 = \left(M_p(P)\right)^2.
	$
	
	By Lemma~\ref{lem:gaussian-correlation-bound},
	\begin{equation}
		\label{eq:probe-variance-pre-final}
		\Var_R[Z] \le C_Z\log n \cdot \left(M_p(P)\right)^2 \cdot \sum_{q=1}^p
		\left(\frac{n}{b_m\tau}\right)^q,
	\end{equation}
	for a sufficiently large constant $C_Z=C_Z(p,d)$.
	Since $n=\tau^{p+1}$ and $b_m=C_b\eps^{-2}\tau^p\log n$, we have
	$
	\frac{n}{b_m\tau} = \frac{\eps^2}{C_b\log n}.
	$
	Choosing the constant $C_b$ in $b_m$ sufficiently large, \eqref{eq:probe-variance-pre-final} becomes
	\begin{equation}
		\label{eq:probe-variance-final}
		\Var_R[Z] \le C_Z\log n \cdot \left(M_p(P)\right)^2 \cdot \sum_{q=1}^p \left(\frac{\eps^2}{C_b\log n}\right)^q \le \frac{\eps^2}{2}\left(M_p(P)\right)^2.
	\end{equation}
	
	We next bound the candidate sampling term in~\eqref{eq:law-total-variance}.  Conditioned on
	$R$, the candidate samples are independent uniform indices in $[n]$, so
	\begin{eqnarray*}
		\Var_I(\widehat M_p\mid R)
		&=& \Var_I\left(\frac{n}{b_m}\sum_{a \in [b_m]} Z_{I_a} \biggm| R\right) =
		\left(\frac{n}{b_m}\right)^2 \sum_{a \in [b_m]}\Var_I(Z_{I_a}\mid R) \\
		&\le& \left(\frac{n}{b_m}\right)^2
		\sum_{a \in [b_m]}\E_I[Z_{I_a}^2\mid R] \\
		&=& \left(\frac{n}{b_m}\right)^2 b_m \cdot \frac1n\sum_{i\in[n]} Z_i^2 
		= \frac{n}{b_m}\sum_{i\in[n]} Z_i^2.
	\end{eqnarray*}
	Taking expectation over the probe samples, and using
	$S_{ii}=\sum_uK_\sigma(p_i,p_u)^2\le D_i$, we get
	\begin{equation*}
		\E_R[Z_i^2] \stackrel{\text{by \eqref{eq:Z}}}{\le} \left(D_i^2+\frac{n}{b_m}D_i\right)^p \le
		C'_Z\sum_{q=0}^p \left(\frac{n}{b_m}\right)^qD_i^{2p-q}.
	\end{equation*}
	Therefore, by Lemma~\ref{lem:density-power-sums},
	\begin{eqnarray*}
		\E_R\left[\Var_I(\widehat M_p\mid R)\right]
		&\le& C_M \frac{n}{b_m}\log n \left(M_p(P)\right)^2
		\sum_{q=0}^p \left(\frac{n}{b_m}\right)^q\frac1{\tau^{q+1}} \\
		&=& C_M \log n\,\left(M_p(P)\right)^2
		\left(\frac{n}{b_m\tau}\right) \sum_{q=0}^p \left(\frac{n}{b_m\tau}\right)^q \\
		&=& C_M \log n\,\left(M_p(P)\right)^2
		\left(\frac{\eps^2}{C_b\log n}\right) \sum_{q=0}^p \left(\frac{\eps^2}{C_b\log n}\right)^q,
	\end{eqnarray*}
	where $C_M=C_M(p,d)$ is a sufficiently large constant.
	Again, for a sufficiently large constant $C_b$, this is at most
	\begin{equation}
		\label{eq:candidate-variance-final}
		\E_R\left[\Var_I(\widehat M_p\mid R)\right] \le \frac{\eps^2}{2}\left(M_p(P)\right)^2.
	\end{equation}
	Combining~\eqref{eq:probe-variance-final} and~\eqref{eq:candidate-variance-final} proves the lemma.
\end{proof}

\begin{proof}[Proof of Theorem~\ref{thm:Fp}]
	For one copy of the estimator, Lemmas~\ref{lem:density-moment-unbiased} and~\ref{lem:density-moment-variance}, together with Chebyshev's inequality, imply that
	$$
	\Pr\left[\abs{\widehat M_p-M_p(P)}>3\eps M_p(P)\right] \le \frac{1}{9}.
	$$
	Run $O(\log(1/\delta))$ independent copies and return the median.  This boosts the success probability to at least $1-\delta$.  Rescaling the accuracy parameter by a constant factor gives the stated $(1\pm\eps)$ guarantee.
	
	Each copy of the estimator stores $b_m$ candidate samples and $p b_m$ probe samples, where
	$
	b_m=O_{p,d}(\eps^{-2}n^{p/(p+1)}\log n).
	$
	Since $p$ is fixed, the space bound follows after multiplying by
	$O(\log(1/\delta))$ independent repetitions.  
\end{proof}

\subsection{Removing the Assumption of Known Stream Length}
\label{sec:unknown-length}

The algorithm presented above assumes that $n$ is known in
advance, since the sample budget $b_m$ is chosen as a function of $n$.
We now remove this assumption without changing the estimator or its
asymptotic space complexity. The only modification is to the implementation
of the candidate and probe samplers.

We run the median amplification above with failure probability
$\delta/2$, and hence set
$
\kappa_m = C_\kappa \log \frac{2}{\delta}.
$
The remaining failure probability $\delta/2$ is allocated to the sampling
construction below. 

Set $\alpha := \frac{p}{p+1}$. Partition the stream positions into dyadic time blocks
$$
\mathcal{J}_j := \{2^j,2^j+1,\ldots,2^{j+1}-1\}, \qquad j=0,1,2,\ldots .
$$
The final block may be only partially filled.
When the first point in block $\mathcal{J}_j$ arrives, initialize
$s_j := C_s\eps^{-2}\kappa_m 2^{\alpha j}(j+1)$ independent reservoir samplers, where $C_s=C_s(p,d)$ is a sufficiently large constant. These samplers receive only the points belonging to $\mathcal{J}_j$. Consequently, after the stream ends, their
outputs are independent uniform samples, with replacement, from
the points in $\mathcal{J}_j$.

Maintain the current stream length using a counter. After the stream ends,
set $\tau := n^{1/(p+1)}$ and $b_m := C_b\eps^{-2}\tau^p\log n$. The higher-moment estimator requires $Q := (p+1)\kappa_m b_m$ independent uniform samples from the final stream: $b_m$ candidate samples and $p \cdot b_m$ probe samples for each of the $\kappa_m$ repetitions.

For every $j$, let $\nu_j := |\mathcal{J}_j\cap[n]|$ be the length of the $j$-th block. Independently for each $t\in[Q]$, draw a block label $J_t$ according to
$\Pr[J_t=j] = {\nu_j}/{n}$. Let $\xi_j := |\{t\in[Q]:J_t=j\}|$
be the number of requested samples from block $\mathcal{J}_j$. If
$\xi_j>s_j$ for some $j$, declare a sampling overflow and return an
arbitrary value. Otherwise, assign the first $\xi_j$ reservoir outputs
from block $\mathcal{J}_j$ to the sample slots whose label is $j$.
Finally, use any fixed bijection between these $Q$ slots
and the candidate and probe samples
$$
\{S_a^{(r)}:r\in[\kappa_m],\,a\in[b_m]\} \ \cup\ \{R_{\ell,k}^{(r)}: r\in[\kappa_m],\,\ell\in[p],\,k\in[b_m]\}.
$$
Conceptually, we try to merge the reservoir samples maintained separately for different blocks into $Q$ samples from the entire stream.

We first show that, in the absence of an overflow, this procedure produces
exactly the samples required by the analysis above. Couple every block
$\mathcal{J}_j$ with an infinite sequence of independent uniform samples
from that block, and regard the stored samples  as the first
$s_j$ elements of this sequence. For any sample slot $t$ and any stream
position $i\in\mathcal{J}_j\cap[n]$, the probability that slot $t$
selects position $i$ is
$$
\Pr[J_t=j]\cdot\frac{1}{\nu_j} = \frac{\nu_j}{n}\cdot\frac{1}{\nu_j} = \frac{1}{n}.
$$
Because the block labels and all block sample sequences are independent, the resulting stream positions are independent and uniform in $[n]$. Therefore, the conceptual construction produces candidate and probe samples with the same joint distribution as those used in the estimator and variance analysis above. On the event that no block overflows, the actual implementation agrees exactly with this conceptual construction.

It remains to bound the probability of overflow. Let
$j_{\max}:=\lfloor\log_2 n\rfloor$. For each $j\leq j_{\max}$, the random
variable $\xi_j$ is binomial with mean
\begin{eqnarray*}
	\lambda_j :=\E[\xi_j]
	&=& Q \cdot \frac{\nu_j}{n} \\
	&\leq& C_0\eps^{-2}\kappa_m
	n^{\alpha-1}2^j\log n \\
	&\leq& C_0\eps^{-2}\kappa_m 2^{\alpha j}(j_{\max}+1)
	2^{-(j_{\max}-j)/(p+1)} \\
	&\leq& C_1\eps^{-2}\kappa_m 2^{\alpha j}(j+1),
\end{eqnarray*}
where $C_0,C_1$ depend only on $p$ and $d$. To justify the last inequality,
write $r=j_{\max}-j$. Then $j_{\max}+1=j+r+1\leq (j+1)(r+1)$, while
$\sup_{r\geq 0}(r+1)2^{-r/(p+1)} = O_p(1)$ for fixed $p$.

Choose $C_s$ sufficiently large that
$s_j\geq 4\lambda_j$ for every $j$. A Chernoff bound then gives
$$
\Pr[\xi_j>s_j] \leq \exp(-c s_j) \leq \frac{\delta}{2^{j+2}},
$$
where the last inequality follows by increasing $C_s$ and $C_\kappa$ if necessary. A union bound over all dyadic blocks gives that the probability that some sample pool overflows is at most $\delta/2$.

Finally, at any time during block $\mathcal{J}_{j_{\max}}$, the number of
stored reservoir samples is
\begin{eqnarray*}
	\sum_{j=0}^{j_{\max}}s_j
	&=& O_{p,d}\left(\eps^{-2}\kappa_m\sum_{j=0}^{j_{\max}}
	2^{\alpha j}(j+1)\right) \\
	&=& O_{p,d}\left(\eps^{-2}\kappa_m 2^{\alpha j_{\max}}(j_{\max}+1)
	\right) \\
	&=& O_{p,d}\left(\eps^{-2}n^{1-1/(p+1)} \log n \log\frac{1}{\delta}\right).
\end{eqnarray*}
This quantity dominates the total space usage.

Let $E$ denote the event that no sample pool overflows. Conditioned on $E$, the actual implementation agrees exactly with the conceptual estimator, and thus the failure probability in this case is at most $\delta/2$. Since $\Pr[E^c]\le \delta/2$, the overall failure probability is at most
$\delta$.

\subsection{Lower Bound}
\label{sec:Fp-lb}

In this section, we give a lower bound for one-dimensional Gaussian density moments by performing a reduction from the classical frequency-moment problem.  Recall that, for a frequency vector $f=(f_1,\ldots,f_N)$, its $k$-th frequency moment is $F_k(f):=\sum_{a\in[N]} f_a^k$. For every fixed integer $k>2$, the classical lower bound for frequency moments states that, for any $\eps \in [C_k N^{-1/k}, 1/10]$ (for a sufficiently large constant $C_k$), any one-pass randomized algorithm that returns a $(1 + \eps)$-approximation to $F_k$ with constant success probability requires
$
\Omega_k\left(\frac{\eps^{-2}N^{1-2/k}}{\log N}\right)
$
bits of space~\cite{Ganguly12}.

\begin{theorem}
	\label{thm:Fp-lb}
	Fix an integer $p>1$.  For $C_p n^{-1/(p+1)} \le \eps\le 1/20$, where $C_p$ is a sufficiently large constant, any one-pass streaming algorithm that, on every stream of at most $n$ points on the real line, returns a $(1+\eps)$-approximation to $M_p(P)$ with probability at least $2/3$ requires
	$
	\Omega_p\left(\frac{\eps^{-2}n^{1-2/(p+1)}}{\log n}\right)
	$
	bits of space.  
\end{theorem}

\begin{proof}
	Set $k:=p+1$. To distinguish it from the frequency stream used in the reduction, we refer to the stream for $M_p$ as the \emph{Gaussian point stream}. Throughout this subsection, we fix the Gaussian kernel bandwidth to $\sigma=1$.
	
	We reduce the insertion-only $F_k$ problem to estimating
	$M_p$.  We use $N$ for the universe size of the frequency stream,	reserving $n$ for the upper bound on the length of the Gaussian point stream. The hard instances underlying the classical $F_k$ lower bound can be chosen to
	have $O(N)$ updates.  More specifically, in the standard reduction the frequency	vector has entries in $\{0,1\}$ except for at most one coordinate of frequency
	$O(\eps N^{1/k})$, together with one additional spike of size $N^{1/k}$; hence the resulting insertion-only stream has length at most $2N$ for sufficiently large $N$.  We therefore set $N:=\left\lfloor\frac n2\right\rfloor$,	so that every frequency stream maps to a Gaussian point stream of length at most $n$.
	
	Let $f=(f_1,\ldots,f_N)$ be the final frequency vector and let
	$\ell:=\sum_{a\in[N]}f_a\le n$ be its stream length.  Choose
	$\eta:=\frac{\eps}{100p}$ and $L:=\sqrt{2\ln\frac n\eta}$.
	For each occurrence of item $a\in[N]$, insert the point
	$z_a:=aL\in\mathbb R$. Thus, the $f_a$ occurrences of item $a$ become $f_a$ coincident Gaussian points at location $z_a$.  Distinct item locations are separated by at least $L$, and	therefore
	$$
	K_1(z_a,z_b) = \exp\left(-\frac{|a-b|^2L^2}{2}\right) \le \exp\left(-\frac{L^2}{2}\right) \le \frac{\eta}{n}
	$$
	for all $a\ne b$.
	
	Every point placed at $z_a$ has the same Gaussian density, which we denote by $D(a)$.  It satisfies
	$$
	D(a) = f_a+\sum_{b\ne a}f_bK_1(z_a,z_b) = f_a+r_a,
	$$
	where $0\le r_a\le \ell\cdot\frac{\eta}{n}\le\eta$. 
	For every active item $a$, we have $f_a\ge1$, and hence
	$$
	f_a\le D(a)\le f_a+\eta\le(1+\eta)f_a.
	$$
	Since there are $f_a$ stream points at $z_a$, it follows that
	$
	M_p(P)=\sum_{a\in[N]}f_aD(a)^p.
	$
	Therefore,
	$$
	F_{p+1}(f)=\sum_{a\in[N]}f_a^{p+1} \le M_p(P) \le (1+\eta)^pF_{p+1}(f).
	$$
	
	For the chosen value of $\eta$ and $\eps\le 1/20$, we have
	$(1+\eta)^p\le e^{\eps/100}\le 1+\frac{\eps}{99}$.
	Hence, if $\widehat M$ is a $(1+\eps)$-approximation to $M_p(P)$, then
	$$
	(1-\eps)F_{p+1}(f) \le \widehat M \le (1+\eps)\left(1+\frac{\eps}{99}\right)F_{p+1}(f) \le (1+2\eps)F_{p+1}(f),
	$$
	so $\widehat M$ is a $(1+2\eps)$-approximation to $F_{p+1}(f)$.  The success probability can be increased from $2/3$ to the constant required by the classical frequency-moment lower bound using a constant number of parallel repetitions.  Since $2\eps\le 1/10$, and $2\eps\ge C_{p+1}N^{-1/(p+1)}$ when $C_{p+1}$ is sufficiently large, the lower bound for $F_{p+1}$ applies with accuracy parameter $2\eps$ and gives
	$$
	\Omega_p\left(\frac{(2\eps)^{-2}N^{1-2/(p+1)}}{\log N}\right) = \Omega_p\left(\frac{\eps^{-2}n^{1-2/(p+1)}}{\log n}\right)
	$$
	bits, which proves the theorem.
\end{proof}
	
	\section{Conclusion}
	
	This paper develops one-pass streaming algorithms for the diversity index and density moments under Gaussian kernel similarity. Several questions remain open. For higher density moments, a small polynomial gap remains between the upper and lower bounds. For the diversity index, our lower bound shows that a dimension-dependent polylogarithmic amount of bit space is necessary; closing the gap between the exponents in the upper and lower bounds remains open. More broadly, it would be interesting to characterize the kernels or geometric similarity functions that admit sublinear-space streaming algorithms.

	\section*{Use of AI assistance}
	The research questions, algorithms, proof strategies, and substantive mathematical arguments in this paper were developed by the author. ChatGPT (GPT-5.5 and GPT-5.6) was used for language editing, proofreading, and to assist with routine algebraic manipulations in the Gaussian-kernel and variance analyses. The resulting text and calculations were verified by the author, who takes full responsibility for the paper’s contents.

%	\bibliographystyle{alpha}
%	\bibliography{paper}
	
\newcommand{\etalchar}[1]{$^{#1}$}

%	\appendix
%	\input{appendix}

\section{Omitted Proofs}

\subsection{Proof of Lemma~\ref{lem:fixed-weight-index}}
\label{sec:proof-lem-index}

\begin{proof}
	Let $X$ be uniform over the vectors in $\{0,1\}^m$ of Hamming weight $m/2$, and let $I$ be uniform in $[m]$ and independent of $X$.  Let $\mathcal R$ denote the public randomness and let $\mathcal M=\mathcal M(X,\mathcal R)$ be Alice's message.  If Bob recovers $X_I$ with error probability at most $1/3$, then the binary form of Fano's inequality gives
	$$
	H(X_I\mid \mathcal M,\mathcal R,I)\le h_2(1/3),
	$$
	where $h_2$ is the binary entropy function.  Since $I$ is uniform and independent of $(X,\mathcal M,\mathcal R)$, we have
	$$
	\sum_{j \in [m]} H(X_j\mid \mathcal M,\mathcal R) = m H(X_I\mid \mathcal M,\mathcal R,I).
	$$
	By subadditivity of conditional entropy,
	$$
	H(X\mid \mathcal M,\mathcal R) \le \sum_{j \in [m]}  H(X_j\mid \mathcal M,\mathcal R) \le m h_2(1/3).
	$$
	On the other hand,
	$$
	H(X)=\log\binom{m}{m/2}=m-O(\log m).
	$$
	It follows that
	$$
	I(X;\mathcal M\mid \mathcal R) = H(X)-H(X\mid \mathcal M,\mathcal R) = \Omega(m).
	$$
	The message length is at least $I(X;\mathcal M\mid \mathcal R)$, proving the lemma.
\end{proof}

\subsection{Proof of Lemma~\ref{lem:SS-continuous-feature}}
\label{sec:Gaussian-identity}

\begin{proof}
	Completing the square gives
	$$
	\norm{z-x}_2^2+\norm{z-y}_2^2 = 2\norm{z-\frac{x+y}{2}}_2^2+\frac{\norm{x-y}_2^2}{2}.
	$$
	Therefore,
	\begin{equation*}
		\int_{\R^d}q_x(z)q_y(z) dz = \left(\frac{2}{\pi\sigma^2}\right)^{d/2}
		\exp\left(-\frac{\norm{x-y}_2^2}{2\sigma^2}\right) \int_{\R^d}
		\exp\left(-\frac{2\norm{z-(x+y)/2}_2^2}{\sigma^2}\right) dz.
	\end{equation*}
	The remaining integral is independent of $x$ and $y$.  With $u=z-(x+y)/2$ and the standard Gaussian integral
	$$
	\int_{\R^d}e^{-a\norm{u}_2^2} du=\left(\frac{\pi}{a}\right)^{d/2},
	\qquad (a>0)
	$$
	we obtain
	$$
	\int_{\R^d}\exp\left(-\frac{2\norm{u}_2^2}{\sigma^2}\right) du = \left(\frac{\pi\sigma^2}{2}\right)^{d/2}.
	$$
	The normalizing constants cancel, and hence
	$$
	\ip{q_x}{q_y}_{L_2(\R^d)} = \exp\left(-\frac{\norm{x-y}_2^2}{2\sigma^2}\right) = K_\sigma(x,y).
	$$
	Finally,
	\begin{eqnarray*}
		\norm{Q}_{L_2(\R^d)}^2
		&=& \ip{Q}{Q}_{L_2(\R^d)}
		=\left\langle \sum_{i \in [n]} q_{p_i},\sum_{j \in [n]} q_{p_j}\right\rangle_{L_2(\R^d)} \\
		&=& \sum_{i \in [n]}\sum_{j \in [n]} \ip{q_{p_i}}{q_{p_j}}_{L_2(\R^d)} \\
		&=& \sum_{i \in [n]}\sum_{j \in [n]} K_\sigma(p_i,p_j)
		=\SUM(P).
	\end{eqnarray*}
\end{proof}

\subsection{Proof of Claim~\ref{lem:discrete-gaussian-tail}}
\label{app:discrete-gaussian-tail}

Recall that
$$
q_u(z) = \left(\frac{2}{\pi\sigma^2}\right)^{d/4} \exp\left(-\frac{\|z-u\|_2^2}{\sigma^2}\right),
$$
and that $\psi(u)_t=h^{d/2}q_u(ht)$, and
$
\phi(u)_t = \psi(u)_t \mathbf{1}\!\left[\|ht-u\|_\infty\leq L\right].
$

Under our choices of parameters, $h=\frac{\sigma}{C_h\sqrt{\Lambda_\eta}}$,
$L=C_L\sigma\sqrt{\Lambda_\eta}$, and $\Lambda_\eta=C_\Lambda\log\frac{dn}{\eta}$.
Since $0<\eta\leq 1/2$ and $n,d\geq 1$, by choosing the constants
sufficiently large we may assume that $h\leq\sigma$ and $L\geq\sigma$.

Define the normalized one-dimensional Gaussian
$$
g_\sigma(s) := \left(\frac{2}{\pi\sigma^2}\right)^{1/2} \exp\left(-\frac{2s^2}{\sigma^2}\right).
$$
Then
$$
q_u(z)^2 = \prod_{r=1}^d g_\sigma(z_r-u_r), \qquad \text{and} \quad \int_{\mathbb R}g_\sigma(s) ds=1.
$$
For $a\in\mathbb R$, define
$$
\mathsf{M}(a) := h\sum_{m\in\mathbb Z}g_\sigma(hm-a)
$$
and
$$
\mathsf{T}_L(a) := h\sum_{\substack{m\in\mathbb Z\\|hm-a|>L}} g_\sigma(hm-a).
$$

We first establish bounds on these one-dimensional shifted-lattice sums.
The function $g_\sigma$ is even and nonincreasing on $[0,\infty)$. On
either half-line, the points of the shifted lattice $h\mathbb Z-a$ are
spaced by $h$. Comparing all but the point closest to the origin with the
integral over the preceding interval of length $h$ gives
$$
\mathsf{M}(a) \leq 2h g_\sigma(0) + 2\int_0^\infty g_\sigma(s) ds =	1+2h g_\sigma(0).
$$
Since $h\leq\sigma$, $\mathsf{M}(a) \leq 1+2\sqrt{\frac{2}{\pi}} <3$.
The same comparison, now applied beyond distance $L$, gives
$$
\mathsf{T}_L(a) \leq 2h g_\sigma(L) + 2\int_L^\infty g_\sigma(s) ds.
$$
Moreover,
$$
\int_L^\infty \exp\left(-\frac{2s^2}{\sigma^2}\right) ds \leq \frac{1}{L} \int_L^\infty s\exp\left(-\frac{2s^2}{\sigma^2}\right) ds = \frac{\sigma^2}{4L} \exp\left(-\frac{2L^2}{\sigma^2}\right).
$$
Using $h\leq\sigma$ and $L\geq\sigma$, we therefore obtain
\begin{equation*}
	\mathsf{T}_L(a) \leq
	2\sqrt{\frac{2}{\pi}}\frac{h}{\sigma}
	\exp\left(-\frac{2L^2}{\sigma^2}\right) +
	\frac{1}{2}\sqrt{\frac{2}{\pi}}\frac{\sigma}{L}
	\exp\left(-\frac{2L^2}{\sigma^2}\right) \leq
	2\exp\left(-\frac{2L^2}{\sigma^2}\right).
\end{equation*}
Both estimates hold uniformly in the shift $a$.

Because $\phi(u)$ is obtained from $\psi(u)$ by deleting the coordinates
outside the truncation region, we have
\begin{equation*}
	\|\psi(u)-\phi(u)\|_2^2 = h^d \sum_{\substack{t\in\mathbb Z^d\\	\|ht-u\|_\infty>L}} q_u(ht)^2 =h^d\sum_{\substack{t\in\mathbb Z^d\\	\|ht-u\|_\infty>L}} \prod_{r \in [d]} g_\sigma(ht_r-u_r).
\end{equation*}
Using
$$
\mathbf{1}\!\left[\|ht-u\|_\infty>L\right] \leq \sum_{r \in [d]} \mathbf{1}\!\left[|ht_r-u_r|>L\right],
$$
and then separating the coordinate sums, we obtain
\begin{equation*}
	\|\psi(u)-\phi(u)\|_2^2 \leq \sum_{r \in [d]}	\mathsf{T}_L(u_r) \prod_{\substack{s \in [d] \\s\neq r}} \mathsf{M}(u_s)  \leq 2d 3^{d-1} \exp\left(-\frac{2L^2}{\sigma^2}\right).
\end{equation*}
Consequently,
$$
\|\psi(u)-\phi(u)\|_2 \leq \sqrt{2d 3^{d-1}}  \exp\left(-\frac{L^2}{\sigma^2}\right).
$$

Substituting $L=C_L\sigma\sqrt{\Lambda_\eta}$ and
$\Lambda_\eta=C_\Lambda\log_2(dn/\eta)$ gives
$$
\|\psi(u)-\phi(u)\|_2 \leq A_d\left(\frac{dn}{\eta}\right)^{-\gamma},
$$
where $A_d:=\sqrt{2d 3^{d-1}}$ and $\gamma:=\frac{C_L^2C_\Lambda}{\ln 2}$.
Since $d$ is fixed and $dn/\eta\geq 2d$, we may choose $C_L$ and
$C_\Lambda$ sufficiently large that $\gamma>1$ and
$A_d(2d)^{-(\gamma-1)}\leq \frac{d}{12}$.
It follows uniformly for all $u\in\mathbb R^d$ that
$$
\|\psi(u)-\phi(u)\|_2 \leq \frac{d}{12}\left(\frac{dn}{\eta}\right)^{-1} = \frac{\eta}{12n}.
$$

\end{document}